\documentclass[11pt,reqno]{amsart}
\usepackage{amsmath,amsthm,amssymb,bm,bbm,dsfont,braket}
\usepackage[mathscr]{eucal}
\usepackage{amsaddr}
\usepackage{upgreek}
\usepackage[left=2cm,right=2cm,top=2cm,bottom=2cm]{geometry}
\usepackage{mathtools}\mathtoolsset{centercolon}
\mathtoolsset{showonlyrefs}

\usepackage{setspace}
\usepackage{graphicx}
\usepackage{verbatim}
\usepackage{float}
\usepackage{placeins}
\usepackage{array}
\usepackage{booktabs}
\usepackage{threeparttable}
\usepackage[update,prepend]{epstopdf}
\usepackage{multirow}
\usepackage{amsfonts,amssymb,dsfont,xfrac,comment}
\usepackage[abs]{overpic}

\usepackage[usenames,dvipsnames]{color}
\usepackage[hidelinks]{hyperref}
\hypersetup{
	unicode=false,          % non-Latin characters in Acrobat’s bookmarks
	pdftoolbar=true,        % show Acrobat’s toolbar?
	pdfmenubar=true,        % show Acrobat’s menu?
	pdffitwindow=false,     % window fit to page when opened
	pdfstartview={FitH},    % fits the width of the page to the window
	pdftitle={My title},    % title
	pdfauthor={Author},     % author
	pdfsubject={Subject},   % subject of the document
	pdfcreator={Creator},   % creator of the document
	pdfproducer={Producer}, % producer of the document
	pdfkeywords={keyword1} {key2} {key3}, % list of keywords
	pdfnewwindow=true,      % links in new PDF window
	colorlinks=true,        % false: boxed links; true: colored links
	linkcolor=Red,          % color of internal links (change box color with linkbordercolor)
	citecolor=ForestGreen,  % color of links to bibliography
	filecolor=Magenta,      % color of file links
	urlcolor=BlueViolet,    % color of external links
}
\usepackage{doi}
\usepackage{url}
\usepackage{caption, subcaption}
\usepackage{enumitem}
\usepackage{shadethm}

\makeatletter
\ifx\@NODS\undefined%

\let\mathbb=\mathds
\else%
\fi
\makeatother

\DeclareMathOperator{\Tr}{Tr}
\DeclareMathOperator{\wt}{wt}

\DeclareMathOperator{\e}{\mathrm{e}}

\newcommand{\si}{{\sigma}}

\newcommand{\be}{{\mathbf e}}

\newcommand{\tr}{\operatorname{Tr}}
\newcommand{\al}{{\alpha}}

\newcommand{\ten}{\otimes}
\newcommand{\pl}{\hspace{.1cm}}

\def\0{{\mathbf{0}}}
\def\1{{\mathbf{1}}}
\def\2{{\mathbf{2}}}
\def\3{{\mathbf{3}}}
\def\4{{\mathbf{4}}}
\def\5{{\mathbf{5}}}
\def\6{{\mathbf{6}}}

\def\7{{\mathbf{7}}}
\def\8{{\mathbf{8}}}
\def\9{{\mathbf{9}}}

\def\be{\begin{equation}}
\def\ee{\end{equation}}
\def\bea{\begin{eqnarray}}
\def\eea{\end{eqnarray}}

\theoremstyle{plain}
\newtheorem{theo}{Theorem} %[section]%[chapter]
\newtheorem{prop}[theo]{Proposition} %[section]
\newtheorem{lemm}[theo]{Lemma} %[section]
\newtheorem{coro}[theo]{Corollary} %[section]
\newshadetheorem{theorem}{Theorem}
\theoremstyle{definition}
\newtheorem*{defn*}{Definition} %[section]

\theoremstyle{remark}
\newtheorem{remark}{Remark}[section]

\newtheorem{cons}{Construction}
\numberwithin{equation}{section}

\makeatletter
\newcommand{\opnorm}{\@ifstar\@opnorms\@opnorm}
\newcommand{\@opnorms}[1]{%
	$\left|\mkern-1.5mu\left|\mkern-1.5mu\left|
	#1
	\right|\mkern-1.5mu\right|\mkern-1.5mu\right|$
}
\newcommand{\@opnorm}[2][]{%
	\mathopen{#1|\mkern-1.5mu#1|\mkern-1.5mu#1|}
	#2
	\mathclose{#1|\mkern-1.5mu#1|\mkern-1.5mu#1|}
}
\makeatother

\begin{document}

%%%%% How to disable amsart.cls to capitalize the article title?
\let\origmaketitle\maketitle
\def\maketitle{
	\begingroup
	\def\uppercasenonmath##1{} % this disables uppercasing title
	\let\MakeUppercase\relax % this disables uppercasing authors
	\origmaketitle
	\endgroup
}
%%%%%%%%%%%%%%%%%%%%%%%%%%%%%%%%

\title{\bfseries \Large{ 
Strong Converse for Privacy Amplification \\ against Quantum Side Information
		}}
		
\author{ \normalsize \textsc{Yu-Chen Shen$^{1}$, Li Gao$^{5}$, and Hao-Chung Cheng$^{1,2,3,4}$}}
\address{\small  	
$^{1}$Department of Electrical Engineering and Graduate Institute of Communication Engineering,\\ National Taiwan University, Taipei 106, Taiwan (R.O.C.)\\
$^{2}$Department of Mathematics, National Taiwan University\\
$^{3}$Center for Quantum Science and Engineering,  National Taiwan University\\
$^{4}$Hon Hai (Foxconn) Quantum Computing Center, New Taipei City 236, Taiwan\\
$^{5}$Department of Mathematics, University of Houston, Houston, TX 77204, USA\\
}

\email{\href{mailto:haochung.ch@gmail.com}{haochung.ch@gmail.com}}
% \email{\href{mailto:gaolimath@gmail.com}{gaolimath@gmail.com}}
\date{\today}
\begin{abstract}
	%\textcolor{red}{Check the order of author names}
	We establish a one-shot strong converse bound for privacy amplification against quantum side information using trace distance as a security criterion. 
	This strong converse bound implies that in the independent and identical scenario, the trace distance exponentially converges to one in every finite blocklength when the rate of the extracted randomness exceeds the quantum conditional entropy. The established one-shot bound has an application to bounding the information leakage of classical-quantum wiretap channel coding and private communication over quantum channels.
	That is, the trace distance between Alice and Eavesdropper's joint state and its decoupled state vanishes as the rate of randomness used in hashing exceeds the quantum mutual information. On the other hand, the trace distance converges to one when the rate is below the quantum mutual information, resulting in an exponential strong converse.
	Our result also leads to an exponential strong converse for entropy accumulation, which complements a recent result by Dupuis [arXiv:2105.05342]. 
	Lastly, our result and its applications apply to the moderate deviation regime. Namely, we characterize the asymptotic behaviors of the trace distances when the associated rates approach the fundamental thresholds with speeds slower than \(O(\sfrac{1}{\sqrt{n}})\).
	
	%We prove an strong converse bound for privacy amplification against quantum side information. Combined with the recent achievability result by Dupuis [arXiv:2105.05342], this shows that the optimal rate of privacy amplification is given by mutual information. Beyond that, we obtain a second order converse bound and the convergence in moderate deviation regime. Our result has applications in entropy accumulation. \\
\end{abstract}
\maketitle
% \tableofcontents

\section{Introduction}
% Literature survey of Privacy Amplification
\emph{Privacy amplification} (also called \emph{randomness extraction} in \cite{SIAM17})
%, first introduced by Bennett, Brassard, and Robert ,
is a vital protocol in classical and quantum cryptography for extracting randomness from a source partially leaked to environment.
Privacy amplification has been widely studied for its applications in security bounds (known as the leftover hash lemma) \cite{Ren05,Tom16,Hay13,Dup21}, random number generation \cite{Hay13}, channel coding \cite{IEEE64,IEEE57},  wiretap channel coding \cite{Hay13,Hay11,Hay2112}, quantum key distribution,
as well as error correction and data compression \cite{Proc465,Proc467,IEEE64,Tsurumaru2021EquivalenceOT}.
Many works have been made on characterizations of information leakage via privacy amplification. The achievability part of privacy amplification against quantum side information has been obtained using security criteria such as trace distance \cite{Ren05,MH333,Tom16,Dup21}, purified distance \cite{TH13,KL21}, and the quantum relative entropy \cite{KL21,YW66}.
On the other hand, a converse bound using purified distance as security criterion has been found in \cite{TH13}. However, a \emph{direct} strong converse analysis via trace distance as the security criterion is still unclear. The main goal of this paper is to establish a large deviation type exponential strong converse bound for privacy amplification against quantum side information using the trace distance as the security criterion.
%Both of our results uses trace distance as the security criterion.

%\subsection{Privacy Amplification with Quantum Side Information}
%Let $X$ be a classical system and $E$ be the environment system.

%In the following, we describe the protocol of privacy amplification.
Throughout the paper, we consider privacy amplification against quantum adversaries.
Suppose Alice and Eve (adversary) share a classical-quantum (c-q) state $\rho_{XE}:= \sum_{x\in\mathcal{X}} p_X(x)|x\rangle \langle x|\otimes \rho_E^x$, where Alice holds classical system $X$ and Eve holds quantum system $E$.
The goal of privacy amplification is for Alice to extract randomness, say on system $Z$, such that it is independent of the quantum side information $E$.
%Consider both Alice and Eve (adversary) having access to system $X$ and $E$. The goal of privacy amplification is to transform a joint state $\rho_{XE}$ to be independent of the access $E$. Whereas this was first introduced for both $X$ and the environment $E$ being classical, it can be extended to a situation where the adversary holds quantum side information instead of classical one. In this case, Alice and Eve hold a classical-quantum (c-q) state $\rho_{XE}$ and the privacy amplification protocol
The conventional protocol of privacy amplification is to apply a random hash function to Alice's system.
In this paper, we adopt the \emph{strongly $2$-universal hash function} as follows.
\begin{defn*}[Strongly $2$-universal hash functions]
	\label{defn:hash}
	A random hash functions $h :\mathcal{X} \to \mathcal{Z}$ is strongly $2$-universal if for all $ x, x'\in \mathcal{X}$ with $x\neq x'$ and $z,z'\in \mathcal{Z}$,
	\begin{align}
		\Pr_{h}\left\{ h(x) = z \; \wedge \; h(x') = z' \right\} = \frac{1}{|\mathcal{Z}|^2}\, .
	\end{align}
\end{defn*}
Namely, the output $h(x)$ for every input $x$ is uniform and pairwise independent.
The implementation of the hash function $h$ is the following linear operation $\mathcal{R}^h_{X\to Z} $ on Alice's system $X$, i.e.~
\begin{align}
	\mathcal{R}^h(\rho_{XE}) &:= \sum_{x\in\mathcal{X}} p_X(x)|h(x)\rangle \langle h(x)|\otimes \rho_E^x\\
	&= \sum_{z\in\mathcal{Z}} |z\rangle\langle z|\otimes\left( \sum_{x\in h^{-1}(z)} p_X(x)\rho_E^x\right).
\end{align}
Alice's goal is to make the extracted randomness close to a uniform distribution $\sfrac{\mathds{1}_Z}{|\mathcal{Z}|}$ and independent of $E$, which can be measured by
the \emph{trace distance} as a security criterion:
\begin{align}
	\varepsilon_\text{PA}%(|\mathcal{Z}|)
	:=\frac{1}{2}\mathds{E}_h\left\| \mathcal{R}^h(\rho_{XE}) - \frac{\mathds{1}_Z}{|\mathcal{Z}|} \otimes \rho_E \right\|_1.
\end{align}

Our first main result is a one-shot strong converse bound for privacy amplification against quantum side information, i.e.~an exponential convergence of $\varepsilon_\text{PA} \to 1$ when $\log|\mathcal{Z}|$ is too large. This result applies to the independent and identical (i.i.d.) scenario, where Alice and Eve now hold $n$-fold product $\rho_{XE}^{\otimes n}$. Combined with the recent achievability result by Dupuis \cite{Dup21}, the following hold
%This result directly applies to the independent and identical (i.i.d.) extension scenario, where Alice and Eve now hold $\rho_{XE}^{\otimes n}$.
%a strong converse bound of the security measure that for any \emph{strongly 2-universal} families $\mathfrak{H}$ of hash function (see Definition \ref{defn:hash}) and all  $\frac{1}{2}\leq \alpha \leq 1$,
%\begin{align}
%	\frac12\mathds{E}_h \left\|  \left(\mathcal{R}_h - \mathcal{U} \right) \rho_{XE} \right\|_1
%	\geq 1 - 4\, \mathrm{e}^{-\frac{1-\alpha}{\alpha} \left( \log |\mathcal{Z}| -  H_{2-\sfrac{1}{\alpha}}^{\downarrow}(X{\,|\,}E)_\rho\right) }.
%\end{align}
%This implies that when the size of $Z$ system is large, the information one can hide from the adversary Eve will be exponentially small.
%applications (QKD protocol, wiretap channel), fully quantum privacy amplifications...
for every blocklength $n\in\mathds{N}$ and $|\mathcal{Z}^n|=e^{nR}$ (Theorem~\ref{theo:sc}),
\begin{align}
\begin{dcases}
	\varepsilon_{\text{PA}} \leq  \mathrm{e}^{-n \sup_{\alpha\in(1,2)}\frac{1-\alpha}{\alpha}  \left( R -  H_{\alpha}^{*}(X{\,|\,}E)_\rho\right) }, &  R < H(X{\,|\,}E)_\rho %\label{eq:Dupuis}
	\\
	\varepsilon_\text{PA} \geq 1 -  4 \, \mathrm{e}^{ -n  \sup\limits_{\alpha\in(\sfrac12,1)} \frac{1-\alpha}{\alpha} \left( R -  H_{2-\sfrac{1}{\alpha}}^{\downarrow}(X{\,|\,}E)_\rho  \right) }, &R > H(X{\,|\,}E)_\rho %\label{eq:sc_iid}
\end{dcases}
\end{align}
where $H_{\alpha}^{*}(X{\,|\,}E)_\rho$ (resp.~$H_{2-\sfrac{1}{\alpha}}^{\downarrow}(X{\,|\,}E)_\rho$) is a sandwiched-- (resp.~Petz--) R\'enyi version of the quantum conditional entropy $H(X{\,|\,}E)_\rho$ (see Section~\ref{sec:notation} for the definitions). Here, the upper bound of the trace distance  $\varepsilon_\text{PA}$  was proved by Dupuis \cite[Theorem 8]{Dup21}, and the lower bound follows from our one-shot strong converse bound, which shows that $\varepsilon_\text{PA}$ converges to $1$ exponentially fast when the rate of the extracted randomness is above $H(X{\,|\,}E)_\rho$.
%(see Section~\ref{sec:notation} for the detailed definitions).
%We highlight that our result holds for any finite blocklength $n\in\mathds{N}$.
% We need to compare our strong converse bound with Tomamichel and Hayashi because they also have one-shot converse bound for purified distance and trace distance (in Marco's book).
% The bound uses the smooth min-entropy, while our bound has a closed-form expression (it holds for all alpha in [1/2,1]). Also, as seen later we have moderate deivation results.

Our result has an application in bounding the information leakage to eavesdropper (Eve) when transmitting message $m\in\{1,\ldots, M\}$ to Bob through a classical-quantum wiretap channel. 
Let $\sigma_{ME}^{\mathcal{C}}$ be the joint classical-quantum (c-q) state between Alice and Eve when Alice employs a random codebook $\mathcal{C}$ (with codewords drawn according to distribution $p_X$). We define $\varepsilon_{\text{leakage}}$ as the trace distance between $\sigma_{ME}^{\mathcal{C}}$ and its decoupled product state as a security index:
\begin{align}
    \varepsilon_{\text{leakage}} :=
    \frac{1}{2} \mathds{E}_{\mathcal{C}\sim p_X }\left\| \sigma_{ME}^{\mathcal{C}} - \frac{\mathds{1}_M}{M}\otimes \sigma_{E}^{\mathcal{C}} \right\|_1 \, .
\end{align}
% Let $\varepsilon_{\text{leakage}}$ be the trace distance between the joint state of Alice's messages and Eve's quantum side information and its decoupled product state. 
Then, we obtain that for any $R$ bits of randomness in hashing used for Alice's secret communication and every coding blocklength $n\in\mathds{N}$, % (Theorems~\ref{theo:wireach} and \ref{theo:wiretap}),
\begin{align}
\begin{dcases}
% \begin{numcases}{}
    \varepsilon_{\text{leakage}} \leq  2\mathrm{e}^{-n \sup_{\alpha\in(1,2)}\frac{1-\alpha}{\alpha}  \left(  I_{\alpha}^{*}(X{\,:\,}E)_\rho - R\right) }, & R> I(X{\,:\,}E)_\sigma  \\
    \varepsilon_{\text{leakage}} \geq 1-5 \mathrm{e}^{-n \sup\limits_{\alpha\in(1/2,1)} \frac{1-\alpha}{\alpha} \left(  I_{2-\sfrac{1}{\alpha}}^{\downarrow}(X{\,:\,}E)_\rho - R\right) }, & R<I(X{\,:\,}E)_\sigma% \end{numcases}
\end{dcases}
\end{align}
where $I_{\alpha}^{*}(X{\,:\,}E)_\sigma$ (resp.~$I_{2-\sfrac{1}{\alpha}}^{\downarrow}(X{\,:\,}E)_\sigma$) is a sandwiched-- (resp.~Petz--) R\'enyi version of the quantum mutual information $I(X{\,:\,}E)_\sigma$ (see Section~\ref{sec:notation}), and $\sigma_{XE} := \sum_{x\in\mathcal{X}} p_X(x)|x\rangle \langle x|\otimes \sigma_E^x$ for each $\sigma_E^x$ being the output at Eve's wiretap channel. We proved the upper bound in Theorem \ref{theo:wireach}. Note that a slightly different upper bound was obtained earlier by Jiawei \textit{et al.} \cite{Hay2112}. The strong converse lower bound (Theorem~\ref{theo:wiretap}) relies on our previous one-shot strong converse of privacy amplification.
These results of c-q wiretap channel coding indicates that if the rate of the randomness in Alice's hashing is above $I(X{\,:\,}E)_\sigma$, then $\varepsilon_{\text{leakage}}\to 0$ exponentially fast. On the other hand, if the rate is below $I(X{\,:\,}E)_\sigma$, then $\varepsilon_{\text{leakage}}\to 1$ exponentially fast, resulting in an \emph{exponential strong converse} \cite{WWY14, MO17, Cheng2021b}.
We remark that similar results applies to bounding the information leakage for private communications over a quantum channel (Corollaries~\ref{coro:private_ach} and \ref{coro:private_sc}) .

Our result also gives an application in strong converse for \emph{entropy accumulation} (EA) \cite{DFR20,Dup21}. The question of EA we ask here is that: given classical side information $X^n_1$ and a global statistical information $T^n_1$, how much uncertainty remains about the classical bit-string $A^n_1$? The security index used here is the trace distance as below:% \cite{DFR20, Dup21}:
\begin{align*}\varepsilon_\text{EA}(w) :=\frac{1}{2}\mathbb{E}_h&\left\|\mathcal{R}^h(\rho_{A^n_1X^n_1E\mid\wt(T^n_1)=w})-\frac{\mathds{1}}{2^{nR}}\otimes\rho_{X^n_1E\mid\wt(T^n_1)=w}\right\|_1. 
\end{align*}
%\begin{IEEEeqnarray}{rCl}
%\varepsilon_\text{QKD}(w) :=\frac{1}{2}\mathbb{E}_h\|\mathcal{R}^h(\rho_{A^n_1X^n_1E|\wt(T^n_1)=w})
%\\
%-\frac{\mathds{1}}{2^{nR}}\otimes\rho_{X^n_1E|\wt(T^n_1)=w}\|_1.
%\end{IEEEeqnarray}
The smaller $\varepsilon_\text{EA}$ means more uncertainty about the variable $A^n_1$. Then, we have (Theorem~\ref{theo:QKD}):
\begin{align}
\begin{dcases}
	\varepsilon_{\text{EA}}(w) \leq  c_{w}\e^{-n\frac{1}{2}\left(\frac{R-f(w)}{V}\right)^2}, & R < f(w)  \\
	\varepsilon_{\text{EA}}(w) \geq 1-c'_{w}\e^{-n\frac{1}{2}\left(\frac{R-f(w)}{V}\right)^2},
	& R > f(w)
\end{dcases}.
\end{align}
Here, $w$ is a parameter on how much global information we can know; $f(w)$ is the \emph{tradeoff function} \cite{DFR20}; $V$ is a constant; and $c_w, c'_w$ are positive constants depending on $w$.
The upper bound was shown in \cite[Theorem 9]{Dup21}, and the lower bound relies on our one-shot strong converse bound.% and an expansion of the quantum conditional R\'enyi entropy \cite{DFR20}.

Lastly, our results extends to the \emph{moderate deviation regime} \cite{CH17, CTT2017}. That is, for every moderate sequence $(a_n)_{n\in\mathds{N}}$ satisfying (i) $ a_n \downarrow 0$, (ii) $ n a_n^2 \uparrow \infty$, we obtain the following asymptotic error behaviors\footnote{Here, by ``$f(n)\lesssim g(n)$" we meant $\lim_{n\to\infty} \frac{1}{n a_n^2} \log f(n) \leq \lim_{n\to\infty} \frac{1}{n a_n^2} \log g(n) $. See Propositions~\ref{prop:moderate_PA}, \ref{prop:moderate_wiretap}, and \ref{prop:moderate_EA} for the precise statements.} 
as $n\to \infty$ (Proposition~\ref{prop:moderate_PA}):
\begin{align}
\begin{dcases}
	\varepsilon_{\text{PA}} \lesssim \mathrm{e}^{-\frac{na_n^2}{2V(X{\,|\,}E)_\rho}  } \to 0 & R = H(X{\,|\,}E)_\rho - a_n,  \\
	\varepsilon_\text{PA} \gtrsim 1-  \mathrm{e}^{-\frac{na_n^2}{2V(X{\,|\,}E)_\rho} } \to 1 & R = H(X{\,|\,}E)_\rho + a_n,
\end{dcases},
\end{align}
where $V(X{\,|\,}E)_\rho$ is the conditional quantum information variance.
Here, the upper bound can be derived based on Dupuis' result \cite[Theorem 8]{Dup21} of error exponent. The lower bound means that even when the rate of the extracted randomness approaches to $H(X{\,|\,}E)_\rho$ from above at a speed slower than $O(\sfrac{1}{\sqrt{n}})$, the trace distance $\varepsilon_{\text{PA}}$ still converges to $1$ asymptotically.
Similar result in the moderate deviation regime also hold for information leakage of c-q wiretap channel coding and entropy accumulation as well (Proposition~\ref{prop:moderate_wiretap}):
\begin{align}
\begin{dcases}
	\varepsilon_{\text{leakage}} \lesssim \mathrm{e}^{-\frac{na_n^2}{2V(X{\,:\,}E)_\sigma}  } \to 0, & R = I(X{\,:\,}E)_\sigma + a_n  \\
	\varepsilon_\text{leakage} \gtrsim 1-  \mathrm{e}^{-\frac{na_n^2}{2V(X{\,:\,}E)_\sigma} } \to 1, & R = I(X{\,:\,}E)_\sigma - a_n
\end{dcases},
\end{align}
where $V(X{\,:\,}E)_\sigma$ is the quantum information variance, and (Proposition~\ref{prop:moderate_EA})
\begin{align}
\begin{dcases}
	\varepsilon_{\text{EA}}(w) \lesssim \e^{-\frac{n a_n^2}{2V}}, & R =  f(w) - a_n  \\
	\varepsilon_{\text{EA}}(w) \gtrsim 1- \e^{-\frac{n a_n^2}{2V}},
	& R = f(w) + a_n
\end{dcases}.
\end{align}

The paper is structured as follows. In the rest of this section we compare our works with existing literature. Section \ref{sec:notation} reviews the necessary background on entropy quantities. In Section \ref{sec:sc}, we prove our main result: a one-shot strong converse for privacy amplification and its $n$-shot extensions. In Section \ref{sec:c-q}, we bound the information leakage in classical-quantum wiretap channel coding. Section \ref{sec:QKD} includes an application to entropy accumulation. Section \ref{sec:moderate} includes moderate deviation analysis of privacy amplification and the applications on wiretap channel and entropy accumulation. We conclude the paper in Section \ref{sec:conclusion}.  We arrange some proofs in Appendix~\ref{sec:proofs}.
%Due to the length limit, we refer to the full version \cite{SGC} \textcolor{red}{Li:add citation for full version} for most of the proofs.
%\textcolor{red}{(Li:Section 3, 4, and 6 can be merged)}

\subsection{Comparison with the existing results}
The well-known \emph{leftover hash lemma} (LHL) of privacy amplification against quantum side information \cite{KMR05, Ren05, Dup21} states that: 
\begin{align}
	\varepsilon_\text{PA} &\leq \e^{ \frac{\alpha-1}{\alpha}\left( \log |\mathcal{Z}| - H^*_\alpha(X{\,|\,}E)_\rho\right) } \quad (\forall \alpha\in (1,2))  \\
	&\leq \e^{ \frac12 \left( \log |\mathcal{Z}| - H^*_2(X{\,|\,}E)_\rho\right) } \\
	&\leq \e^{ \frac12 \left( \log |\mathcal{Z}| - H^*_\infty(X{\,|\,}E)_\rho\right) }.
\end{align}
Our one-shot strong converse bound (Theorem~\ref{theo:sc}) bears a resemblance to the LHL in a complementary way:
\begin{align}
	\varepsilon_\text{PA} &\geq 1 - 4\e^{ \frac{\alpha-1}{\alpha}\left( \log |\mathcal{Z}| - H^\downarrow_\alpha(X{\,|\,}E)_\rho\right) } \quad \forall\alpha\in (\sfrac12,1) \\
	&\geq 1 - 4\e^{ -\frac12 \left( \log |\mathcal{Z}| - H^\downarrow_{\sfrac12}(X{\,|\,}E)_\rho\right) }.
\end{align}
Together, they imply that, for $\varepsilon$-secrete privacy amplification protocols, the maximal number of bits of the extractable uniform randomness ($\log |\mathcal{Z}|$) is bounded as
\begin{align}
	H_\infty^*(X{\,|\,}E)_\rho - 2 \log \frac{1}{\varepsilon}
	\leq \log |\mathcal{Z}| 
	\leq H^\downarrow_{\sfrac12} (X{\,|\,}E)_\rho + 2\log\frac{4}{1-\varepsilon}.
\end{align}

Comparing with the standard converse bounds (e.g.~\cite{TSS+11, TH13, Tom16}), our result gives a \emph{direct} converse analysis to the trace distance as the security criterion, which does not require intermediate steps through analysis via the purified distance \cite{TCR10} and the Fuchs-van de Graaf inequality \cite{FG99}. Moreover, the conditional R\'enyi entropy $H^\downarrow_{\alpha}$ playing as a role of the exponent has a closed-form expression (see Section~\ref{sec:notation} for detailed definition) as opposed to the smooth entropies  \cite{TCR10}. %$H_\infty^\alpha$.
 When considering the i.i.d.~extension of $n$-fold product state $\rho_{XE}^{\otimes n}$, the additivity of $H^\downarrow_{\alpha}$ immediately yields exponential bound on $1-\varepsilon_\text{PA}$ for \emph{every} finite blocklength without appealing to asymptotic expansion via the smooth entropies \cite{TCR10}. Similarly, our result implies strong converse for the entropy accumulation \cite{DFR20} in device-independent quantum key distribution without going through smooth entropies, which has a similar flavor as Ref.~\cite{Dup21} in the achievability part.  Lastly, our result provides a R\'enyi-type entropy to characterize the one-shot operational quantity ($\varepsilon_\text{PA}$), partially answering Dupuis' question raised in \cite{Dup21}.

\section{Notation and Information Quantities} \label{sec:notation}
We denote the $[M]:= \{1,\ldots, M\}$ for any integer $M\in\mathds{N}$.
We denote $\mathcal{B(H)}$ as the space of bounded linear operators on a Hilbert space $\mathcal{H}$, and $\mathcal{B}_{\geq 0}(\mathcal{H})$ as the set of positive (semi-definite) operators. For an operator $H\in \mathcal{B(H)}$, the Schatten-$p$ norm is defined as
\begin{align}
    \|H\|_p := \Tr\left[ \left( H^\dagger H \right)^{\sfrac{p}{2}} \right],
\end{align}
where $\Tr$ is the standard matrix trace. The set of density operators (positive with unit trace) is denoted as $\mathcal{S(H)}$.
%For $p\geq 1$, the Schatten $p$-norm for an operator $M$ is
%\begin{align}
%	S_p(\mathcal{H}) &:= \left\{ M \in \mathcal{B(H)} : \left\|M\right\|_{S_p(\mathcal{H}) } < \infty  \right\}; \\
%$\left\|M\right\|_{S_p(\mathcal{H}) } = \left( \Tr\left[ | M|^p \right] \right)^{\sfrac1p}$.
%\end{align}
%and the Schatten $p$-class $S_p(\mathcal{H})$ is the space of all operators with finite Schatten $p$-norm.
%We will often shorthand $\|\cdot\|_p \equiv \|\cdot\|_{S_p(\mathcal{H}) }$ if the underlying Hilbert space is clear.
We use $\text{supp}(\cdot)$ to stand for the support of a function or the support of an operator.

Recall that for $\alpha\in(0,\infty)\backslash 1$, the order-$\alpha$ Petz--R\'enyi divergence $D_\alpha$ \cite{Pet86} is defined as
and the sandwiched R\'enyi divergence $D^*_\alpha$ \cite{MDS+13,WWY14} are defined  as
\begin{align}
D_\alpha(\rho\|\sigma) &:=\frac{1}{\al-1}\log\tr\left[\rho^{\alpha}\sigma^{1-\alpha}\right]\ ; \\ 
D_\alpha^*(\rho\|\sigma)&:=\frac{1}{\al-1}\log\left\|{\sigma^{\frac{1-\alpha}{2\alpha}}\rho\si^{\frac{1-\alpha}{2\alpha}}}\right\|_{\al}^\al \ ,
\end{align}
where $\rho \in \mathcal{S(H)}, \sigma \in \mathcal{B}_{\geq 0}(\mathcal{H})$ and $\text{supp}(\rho)\subseteq \text{supp}(\sigma)$.
%\end{align}
Note that when $\alpha\to 1$, both R\'enyi divergences converge to the quantum relative entropy \cite{Ume62} $D(\rho\|\sigma) := \Tr\left[ \rho (\log \rho - \log \sigma )\right]$ (see e.g. \cite[Lemma 3.5]{MO17}), i.e.
\begin{align}\lim_{\alpha\to 1} D_\alpha(\rho\|\sigma)=\lim_{\alpha\to 1} D_\alpha^*(\rho\|\sigma)=D\left(\rho\|\sigma\right).\label{eq:limit}\end{align}
It is well-known that both $\alpha \mapsto D_\alpha$ and $\alpha \mapsto D_\alpha^*$ are monotone increasing on $(0,\infty)$ (see e.g.~\cite[Lemma 3.12]{MO17}).

%where $\norm{x}{\al}=\tr(|x|^{\al})^{\frac{1}{\al}}$ is the Schatten $\al$-norm.
\noindent For a classical-quantum state $\rho_{XE} = \sum_{x\in\mathcal{X}} p_X(x) |x\rangle \langle x| \otimes \rho_E^x$, %the \emph{order-$\alpha$ sandwiched R\'enyi information} $I_\alpha^{*}$ and the \emph{order-$\alpha$ sandwiched Augustin information} $\breve{I}_\alpha^{*}$ as:
%\begin{align}
%	\begin{split} \label{eq:sandwiched_Renyi}
%	I_\alpha^{*} \left( X;B \right)_\rho &:= \inf_{\sigma_B\in\mathcal{S}(\mathcal{H}_B)} D_\alpha^*\left( \rho_{XB} \| p_X \otimes \sigma_B \right);
%		\\
%		&= \inf_{\sigma_B\in\mathcal{S}(\mathcal{H}_B)} \frac{\alpha}{\alpha-1} \log \left( \sum_{x\in\mathcal{X}} p_X(x) \left\| \sigma_B^{\frac{1-\alpha}{2\alpha} }  \rho_B^x \sigma_B^{\frac{1-\alpha}{2\alpha}} \right\|_\alpha^\alpha \right)^{\frac{1}{\alpha}};
%	\end{split} \\
%	\begin{split} \label{eq:sandwiched_Augustin}
%	\breve{I}_\alpha^{*} \left( X;B \right)_\rho &:= \inf_{\sigma_B\in\mathcal{S}(\mathcal{H}_B)} \sum_{x\in\mathcal{X}} p_X(x) D_\alpha^*\left( \rho_B^x \| \sigma_B \right),
%		\\
%		&= \inf_{\sigma_B\in\mathcal{S}(\mathcal{H}_B)} \frac{\alpha}{\alpha-1}  \sum_{x\in\mathcal{X}} p_X(x) \log \left\| \sigma_B^{\frac{1-\alpha}{2\alpha} }  \rho_B^x \sigma_B^{\frac{1-\alpha}{2\alpha}} \right\|_\alpha.
%	\end{split}
%\end{align}
%where the infimum $\si_B$ is taken over all densities on $B$.
%Similarly,
we define the following Petz-type and sandwiched type conditional entropy and mutual information:
\begin{align} 
	H^{\downarrow}_{\alpha}(X{\,|\,}E)_\rho
	:=  -D_\alpha\left( \rho_{XE} \| \mathds{1}_X\otimes \rho_E \right) \label{eq:Petz_Renyi_down_con}\ ,\ &H^*_{\alpha}(X{\,|\,}E)_{\rho} = - \inf_{\sigma_E\in \mathcal{S}(\mathcal{H}_E)}D^*_{\alpha}(\rho_{XE}\|\mathds{1}_X\otimes \sigma_E); \\
	I^{\downarrow}_{\alpha}(X{\,:\,}E)_\rho
	:=  D_\alpha\left( \rho_{XB} \| \rho_X \otimes \rho_E \right)\ ,\ 
&I^*_{\alpha}(X{\,:\,}E)_{\rho} = \inf_{\sigma_E\in \mathcal{S}(\mathcal{H}_E)}D^*_{\alpha}(\rho_{XE}\|\rho_X\otimes \sigma_E)\label{eq:Petz_Renyi_down_in}
	%	\breve{I}^{\downarrow}_{\alpha}(X;B)_\rho
	%&:= \sum_{x\in\mathcal{X}} p_X(x) D_\alpha\left( \rho_{B}^x \|  \rho_B %\right). \label{eq:Petz_Augustin_down}
\end{align}
%\textcolor{blue}{$H_\alpha$ $I_\alpha$}\\
Similar to \eqref{eq:limit}, both R\'enyi quantities converges to the usual quantum conditional entropy and quantum mutual information, i.e.
\begin{align} \label{eq:alpha1}
	\begin{split}
		&\lim_{\alpha\to 1}	H^{\downarrow}_{\alpha}(X{\,|\,}E)_\rho=\lim_{\alpha\to 1}	H^{*}_{\alpha}(X{\,|\,}E)_\rho=H(X{\,|\,}E)_\rho:=-D(\rho_{XE}||1\ten \rho_E);\\
		&\lim_{\alpha\to 1}	I^{\downarrow}_{\alpha}(X{\,:\,}E)_\rho=\lim_{\alpha\to 1}	I^{*}_{\alpha}(X{\,:\,}E)_\rho=I(X{\,:\,}E)_\rho:=D(\rho_{XE}||\rho_X\ten \rho_E).
	\end{split}
\end{align}
%\begin{align} \label{eq:alpha1}
%	\begin{split}
%	&\lim_{\alpha\to 1} I_\alpha^{*} \left( X;B \right)_\rho
%	=\lim_{\alpha\to 1}\breve{I}_\alpha^{*} \left( X;B \right)_\rho\\
%%	=\lim_{\alpha\to 1}\breve{I}^{\downarrow}_{\alpha}(X;B)_\rho \\
%	= &I(X;B)_\rho
%	:= D(\rho_{XB}\|\rho_X\otimes \rho_B).
%	\end{split}
%\end{align}
%\textcolor{red}{Li:Do we need this quantity?} \textcolor{blue}{Yu-Chen: it's used in appendix proof of thm 7, I've put the quantity to appendix}
The \emph{relative entropy variance} $V(\rho\|\sigma)$ is defined by
%\begin{align}
\[
V(\rho\|\sigma): = \Tr(\rho(\log \rho -\log \sigma )^2).
\]
%\end{align}
For a c-q state $\rho_{XE}$, the \emph{conditional information variance} $V(X{\,|\,}E)_\rho$ and the  \emph{mutual information variance}
$V(X{\,:\,}E)_{\rho}$ are defined as
\[
V(X{\,|\,}E)_\rho := V(\rho_{XE} \| \mathds{1}_X \otimes \rho_E)
\ ,\  V(X{\,:\,}E)_{\rho} := V(\rho_{XE} \,\|\,\rho_X\otimes \rho_E)\ .\]
%For $\rho \in \mathcal{S(H)}, \sigma \in \mathcal{B}_{\geq 0}(\mathcal{H})$ and $\text{supp}(\rho)\le \text{supp}(\sigma)$, we define the order-$\alpha$ sandwiched R\'enyi divergence \cite{Mul13,Mark14} as 
%$
%D^*_{\alpha}(\rho||\sigma) := \frac{1}{\alpha-1}\log \tr\left[\left(\sigma^{\frac{1-\alpha}{2}}\rho \sigma^{\frac{1-\alpha}{2}}\right)^{\alpha}\right], \forall \alpha \in (0,\infty)\backslash\{1\}
%$ 
%For c-q state $\rho_{XE} = \sum_{x\in \mathcal{X}}p_X(x)\ket{x}\bra{x}\otimes\rho^x_E$, the sandwiched conditional entropy and mutual information are then defined by:
%\begin{align}
%H^*_{\alpha}(X{\,|\,}E)_{\rho} &= - \inf_{\sigma_E\in \mathds{S}(\mathcal{H}_E)}D^*_{\alpha}(\rho_{XE}||\mathds{1}_X\otimes \sigma_E);\\
%I^*_{\alpha}(X{\,:\,}E)_{\rho} &= \inf_{\sigma_E\in \mathds{S}(\mathcal{H}_E)}D^*_{\alpha}(\rho_{XE}||\rho_X\otimes \sigma_E).
%\end{align}
For two classical systems $\mathcal{X}$ and $\mathcal{Z}$, the perfectly randomizing channel $\mathcal{U}_{X\to Z}$ from $\mathcal{X}$ to $\mathcal{Z}$ is defined as
\begin{align}
	\mathcal{U}(\theta_X) = \frac{\mathds{1}_Z}{|\mathcal{Z}|}\left(\sum_{x}\theta_X(x)\right).
\end{align}
For positive semi-definite operators $A$ and positive definite operator $B$, we use the short notation 
\begin{align}
\frac{A}{B}:= B^{-\frac{1}{2}}AB^{-\frac{1}{2}}
\end{align}
for the \emph{noncommutative quotient}.

\section{Strong Converse for Privacy Amplification} \label{sec:sc}
In \cite[Theorem 8]{Dup21}, the author prove an one-shot achievability bound of privacy amplification: letting $\rho_{XE}=\sum_{x\in \mathcal{X}}p_X(x)\ket{x}\bra{x}\otimes\rho^x_E$ be a classical-quantum state. 
For any strongly $2$-universal hash functions $h: \mathcal{X}\to \mathcal{Z}$, 
the following holds for all $\alpha \in (1,2]$,
\begin{align}
\frac12\mathds{E}_h \left\|  \left(\mathcal{R}^h - \mathcal{U} \right) \rho_{XE} \right\|_1 \leq \e^{\frac{\alpha-1}{\alpha}(\log|\mathcal{Z}|-H^*_{\alpha}(X{\,|\,}E)_{\rho})}.\label{dupuis}
\end{align}
When $\log |\mathcal{Z}| < H^*_{\alpha}(X{\,|\,}E)_{\rho}$, the exponent is positive, meaning that the trace distance exponentially decays.

Our main result in this section is to establish an one-shot strong converse, showing that the trace distance, however, converges to $1$ exponentially fast whenever $\log |\mathcal{Z}| > H^*_{\alpha}(X{\,|\,}E)_{\rho})$. This, it complements Ref.~\cite{Dup21} in the strong converse regime.

\begin{theo}[One-shot strong converse] \label{theo:sc}
	Let $\rho_{XE}=\sum_{x\in \mathcal{X}}p_X(x)\ket{x}\bra{x}\otimes\rho^x_E$ be a classical-quantum state. For any strongly $2$-universal hash function $h: \mathcal{X}\to \mathcal{Z}$, the following holds for all $\alpha \in (\sfrac12, 1)$,
	\begin{align}
		\frac12\mathds{E}_h \left\|  \left(\mathcal{R}^h - \mathcal{U} \right) \rho_{XE} \right\|_1
		\geq  1  -4 \,\mathrm{e}^{-\frac{1-\alpha}{\alpha} \left( \log |\mathcal{Z}| -  H_{2-\sfrac{1}{\alpha}}^{\downarrow}(X{\,|\,}E)_\rho\right) }.
		%, \quad \forall .
	\end{align}
	Here, $H_{2-\sfrac{1}{\alpha}}^{\downarrow}$ is defined in \eqref{eq:Petz_Renyi_down_con}.

	Moreover, the exponent $\sup_{\alpha\in(\sfrac12, 1)}
	\frac{1-\alpha}{\alpha} \big( \log |\mathcal{Z}| -  H_{2-\sfrac{1}{\alpha}}^{\downarrow}(X{\,|\,}E)_\rho\big)$ is positive if and only if $\log|\mathcal{Z}|> H(X{\,|\,}E)_\rho$. 
\end{theo}
Using the additivity of the Petz--R\'enyi divergence, $D_\alpha$, our one-shot result easily applies to the independent and identically distributed (i.i.d.)~case where Alice and Eve holds state $\rho_{XE}^{\otimes n}$ with $|\mathcal{Z}^n| = \e^{nR}$.
Moreover, our result holds for every finite blocklength $n\in\mathds{N}$.
\begin{coro}[Finite-blocklength exponential strong converse] \label{coro:sc}
Let $\rho_{XE}=\sum_{x\in \mathcal{X}}p_X(x)\ket{x}\bra{x}\otimes\rho^x_E$ be a classical-quantum state. Then for every $n\in \mathbb{N}$, the rate $R=\frac{1}{n}\log  |\mathcal{Z}|$ and a strongly $2$-universal hash $h^n: \mathcal{X}^n\to \mathcal{Z}^n$ be a strongly $2$-universal hash function,
\begin{align}	\frac12\mathds{E}_h \left\|  \left(\mathcal{R}^{h^n} - \mathcal{U}^n \right) \rho_{XE}^{\ten n} \right\|_1
\geq 1 -  4 \, \mathrm{e}^{-n \frac{1-\alpha}{\alpha} \left( R -  H_{2-\sfrac{1}{\alpha}}^{\downarrow}(X{\,|\,}E)_\rho  \right) }\ ,\ \ \ \  \alpha\in(\sfrac12, 1),
\end{align}
where $\mathcal{U}^n$ is the perfectly randomizing channel from $\mathcal{X}^n$ to $\mathcal{Z}^n$. 
The above trace distance converges to $1$ exponentially fast for every $n\in\mathds{N}$ when the rate $R > H(X{\,|\,}E)_\rho$.
\end{coro}

% By additivity of $D_\alpha$, Theorem~\ref{theo:sc} immediately implies the strong converse that the trace distance converges to $1$ exponentially fast when $R > H(X{\,|\,}E)_\rho$. %\textcolor{red}{Li:Do we need to state this as a theorem}%, i.e.~
%\begin{align}
%\frac{1}{2}\mathbb{E}_h \left\|({\mathcal{R}^{h^n}}-\mathcal{U}^{\otimes n})(\rho_{XE}^{\otimes n})\right\|_1\geq 
%    1- 4 \mathrm{e}^{ -n  \sup\limits_{\alpha\in(\sfrac12,1)} \frac{1-\alpha}{\alpha} \left( R -  H_{2-\sfrac{1}{\alpha}}^{\downarrow}(X{\,|\,}E)_\rho  \right) }.
%\end{align}
Before proving Theorem~\ref{theo:sc}, we first introduce a Lemma that will be used in the proof.

\label{sec:lemmas}
\begin{lemm}[A trace inequality] \label{lemm:trace}
	For non-zero positive semi-definite operators $K$ and $L$ and any $s\in(0,1)$, the following holds,
	\begin{align}
		\Tr\left[ K (K+L)^{-\sfrac12} L (K+L)^{-\sfrac12} \right] \leq \Tr\left[ K^{1-s} L^s \right].%, \forall s\in (0,1).
	\end{align}
\end{lemm}
We defer the proof of Lemma~\ref{lemm:trace} to Appendix~\ref{sec:proof_trace}. Based on that, we shall now prove Theorem~\ref{theo:sc}.
\begin{proof}[Proof of Theorem~\ref{theo:sc}]
For the ease of notation, we shorthand $p \equiv p_X$, each $\rho_x \equiv \rho_E^x$ and
introduce the notation 
\begin{align}
 \rho_{h(x)E} &:=p(x)\ket{h(x)}\bra{h(x)}\ten \rho_x\,; \\
 \pi_Z &:=\frac{\mathds{1}_{\mathcal{Z}}}{|\mathcal{Z}|}\,.
\end{align}
Then
	\begin{align}
	\mathcal{R}^h(\rho_{XE}) &=\sum_{x\in\mathcal{X}}\rho_{h(x)E}\ ;\\  \mathcal{U}(\rho_{XE})&=\pi_Z\ten \rho_E\pl. 
	\end{align}
	Take the measurement
	\begin{align}
		\Pi = \frac{\mathcal{R}^h(\rho_{XE})}{\mathcal{R}^h(\rho_{XE})+\mathcal{U}(\rho_{XE})}=\frac{\sum_{x}\rho_{h(x)E}}{\sum_{x}\rho_{h(x)E}+\pi_Z\ten \rho_E}\ .
	\end{align}
	Recall the  duality that for positive matrices $A$ and $B$ with $\tr(A)=\tr(B)$ \cite[\S 9]{NC09},
	\begin{align}
		\frac{1}{2}\left\|A- B\right\|_1 = \sup_{0\leq \Pi \leq \mathds{1}} \Tr[\Pi(A-B)].
	\end{align}
	Then, we have
	\begin{align}
		\frac{1}{2}\mathbb{E}_{h} \left\| (\mathcal{R}^h-\mathcal{U})(\rho_{XE})\right\|_1 
		&\geq \mathbb{E}_{h} \Tr\left[\big(\mathcal{R}^h(\rho_{XE})-\mathcal{U})(\rho_{XE})\big) \Pi\right]
		\\
		&=\mathbb{E}_{h} \Tr\left[ \mathcal{R}^h(\rho_{XE})\Pi\right] - \mathbb{E}_{h} \Tr\left[ \mathcal{U}(\rho_{XE})\Pi\right].\label{21}
	\end{align}
	The first term is bounded by
	\begin{align}
		\mathbb{E}_{h} \Tr\left[\mathcal{R}^h(\rho_{XE})\Pi\right]
		&= \mathbb{E}_{h}\sum_{x\in \mathcal{X}}\Tr\left[\rho_{h(x)E}
		\frac{\sum_{x'}\rho_{h(x')E}}{\sum_{x'}\rho_{h(x')E}+\pi_Z\ten \rho_E}\right]\\
		&\geq  \mathbb{E}_{h}\sum_{x\in \mathcal{X}}\Tr\left[\rho_{h(x)E}
		\frac{\rho_{h(x)E}}{\sum_{x'}\rho_{h(x')E}+\pi_Z\ten \rho_E}\right]\\
		&= 1 - \mathbb{E}_{h}\sum_{x\in \mathcal{X}}\Tr\left[\rho_{h(x)E}
		\frac{\sum_{x'\neq x}\rho_{h(x')E}+\pi_Z\ten \rho_E}{\sum_{x'}\rho_{h(x')E}+\pi_Z\ten \rho_E}\right]\\
		&\overset{(a)}{\geq} 1-\sum_{x\in \mathcal{X}}\mathbb{E}_{h}\Tr\left[\rho_{h(x)E}^{1-s}\left(\sum_{x'\neq x}\rho_{h(x')E}+\pi_Z\ten \rho_E\right)^s\right]
		\\
		&\overset{(b)}{\geq} 1-\left(\frac{2}{|\mathcal{Z}|}\right)^s\sum_{x\in \mathcal{X}}\Tr\left[\left(p(x)\rho_x\right)^{1-s}
		\left(\rho_E \right)^s\right], \quad \forall s\in(0,1), \label{eq:sc_1}
	\end{align}
	where (a) follows from Lemma \ref{lemm:trace}
	with $K = \rho_{h(x)E}$ and $L=\sum_{x'}\rho_{h(x')E}+\pi_Z\ten \rho_E$.
	We explain inequality (b) as follows. Recal that the hash function is basically a family of pairwise-independent random variables $h(x),x\in \mathcal{X}$ such that for any $z,z'\in \mathcal{Z}$ and $x'\in\mathcal{X}$,
\begin{align}\text{Pr}(h(x)=z\wedge h(x')=z')=\frac{1}{|\mathcal{Z}|^2}\ . \label{eq:inde}\end{align}
Given $x\in \mathcal{X}$ and $z\in \mathcal{Z}$, we denote $\mathbb{E}_{h|h(x)=z}$ as the expectation of $h(x')$ for each $x'\neq x$ conditional on $h(x)=z$.
Then for each $x\in \mathcal{X}$,
	\begin{align}
		&\mathbb{E}_{h}\tr\left[\rho_{h(x)E}^{1-s}\left(\sum_{x'\neq x}\rho_{h(x')E}+\pi_Z\ten \rho_E\right)^s\right]\\
		 &=\sum_{z}\text{Pr}(h(x)=z) \mathbb{E}_{h|h(x)=z} \tr\left[\left(\rho_{zE}^{1-s}\right)\left(\sum_{x'\neq x}\rho_{h(x')E}+\pi_Z\ten \rho_E\right)^s\right]\\
		&\overset{(a)}{=}\sum_{z}\text{Pr}(h(x)=z) \tr\left[\left(\rho_{zE}^{1-s}\right)\mathbb{E}_{h|h(x)=z}\left[\left(\sum_{x'\neq x}\rho_{h(x')E}+\pi_Z\ten \rho_E\right)^s\right]\right]
		\\ 
		&\overset{(b)}{\le} \sum_{z}\text{Pr}(h(x)=z)\tr\left[\left(\rho_{zE}^{1-s}\right)\left(\mathbb{E}_{h|h(x)=z} \left[\sum_{x'\neq x}\rho_{h(x')E}+\pi_Z\ten \rho_E\right]\right)^s\right]
		\\ 
		&\overset{(c)}{\le} \tr\left[\left(\mathbb{E}_{h(x)}\left[ \rho_{h(x)E}^{1-s}\right]\right)\left(\sum_{x'\neq x} \pi_Z\ten p(x')\rho_{x'}+\pi_Z\ten \rho_E\right)^s\right]
		\\ 		
		&\overset{(d)}{\le} \tr\left[\left(\mathbb{E}_{h(x)}\left[ \rho_{h(x)E}^{1-s}\right]\right)\left(\pi_Z\ten \rho_E+\pi_Z\ten \rho_E\right)^s\right]
		\\ 
		&\overset{}{=} \tr\left[\left(\mathbb{E}_{h(x)} \left[ \ket{h(x)}\bra{h(x)}\ten \left(p(x)\rho_x\right)^{1-s}\right]\right)\left(2\pi_Z\ten \rho_E\right)^s\right]
		\\ 
		&= \left(\frac{2}{|\mathcal{Z}|}\right)^s\tr\left[\left(p(x)\rho_x\right)^{1-s}\left(\rho_E\right)^s\right].
	\end{align}
	Here, (a) follows linearity of trace, 
	(b) follows from Jensen's inequality and the operator concavity of power function $(\cdot)^s$ for $s\in(0,1)$;
	(c) follows from the pairwise independence \eqref{eq:inde} between $h(x)$ and $h(x')$ that $x\neq x'$ and the uniformity of $h$:
	\begin{align}&\mathbb{E}_{h\mid h(x)=z} \left[\rho_{h(x')E}\right]  =
	\mathbb{E}_{h} \left[\rho_{h(x')E}\right]
	= p(x')\pi_Z\ten \rho_{x'}\ ;
	\end{align}
	(d) follows from the operator monotonicity of power function $(\cdot)^s$ for $s\in(0,1)$ and
	\begin{align*}
% 	\mathbb{E}_{h|h(x)=z}\left[\sum_{x'\neq x} \rho_{h(x')E}\right]=
	\sum_{x'\neq x}p(x')\pi_Z\ten \rho_{x'} \leq \sum_{x' \in \mathcal{X} }p(x')\pi_Z\ten \rho_{x'}
	= \pi_Z\ten \rho_E\ .
	\end{align*}
	Hence, inequality~\eqref{eq:sc_1} is proved.
	%and finally (e) uses the uniformity of $h(x)$
	%\begin{align*}%&\mathbb{E}_{h}\rho_{h(x)E}=\mathbb{E}_{h}\ket{h(x)}\bra{h(x)}\ten p(x)\rho_E^x=\pi_Z\ten p(x)\rho_E^x\
	%	\mathbb{E}_{h}\sum_x \rho_{h(x)E}=\pi_Z\ten \sum_{x}p(x)\rho_E^x =\pi_Z\ten \rho_E\ . \end{align*}
	
	The second part can be bounded similarly:
	\begin{align}
		\mathbb{E}_{h} \Tr\left[(\pi_Z\ten \rho_E)\Pi\right]
		&=\mathbb{E}_{h} \Tr\left[(\pi_Z\ten \rho_E)\frac{\sum_{x}\rho_{h(x)E}}{\sum_{x}\rho_{h(x)E}+\pi_Z\ten \rho_E}\right]
		\\
		&=\mathbb{E}_{h} \sum_{x\in\mathcal{X}}\Tr\left[\rho_{h(x)E}\frac{(\pi_Z\ten \rho_E)}{\sum_{x'}\rho_{h(x')E}+\pi_Z\ten \rho_E}\right]
		\\
		&\le \mathbb{E}_{h} \sum_{x\in\mathcal{X}}\Tr\left[\rho_{h(x)E}\frac{(\sum_{x'\neq x}\rho_{h(x')E}+\pi_Z\ten \rho_E)}{\sum_{x'}\rho_{h(x')E}+\pi_Z\ten \rho_E}\right]
		\\
		&\le \sum_{x\in\mathcal{X}}\mathbb{E}_{h} \Tr\left[\rho_{h(x)E}^{1-s} \left(\sum_{x'\neq x}\rho_{h(x')E}+\pi_Z\ten \rho_E\right)^s\right]
		\\
		&\le \left(\frac{2}{|\mathcal{Z}|}\right)^s\sum_{x\in\mathcal{X}}\tr\left[\left(p(x)\rho_{E}^x\right)^{1-s}\left(\rho_E\right)^s\right], \quad \forall s\in(0,1).
	\end{align}
	Therefore, by choosing $\alpha=\frac{1}{1+s} \in (\sfrac{1}{2},1)$, we have
	\begin{align}
		\frac{1}{2}\mathbb{E}_{h} \left\| (\mathcal{R}^h-\mathcal{U})(\rho_{XE})\right\|_1 
		&\geq  1- 2\left(\frac{2}{|\mathcal{Z}|}\right)^s\sum_{x\in \mathcal{X}}\Tr\left[\left(p(x)\rho^x_E\right)^{1-s}(\rho_E)^s\right]\\
		&\geq 1-4\e^{-\frac{1-\alpha}{\alpha}\left(\log|\mathcal{Z}|-H^{\downarrow}_{2-\sfrac{1}{\alpha}}(X{\,|\,}E)\right)}.
	\end{align}
	%where we choose $\alpha=\frac{1}{1+s} \in (\sfrac{1}{2},1)$. 
The positivity of the exponent follows from the monotone decreasing of $\alpha\mapsto H_{2-\sfrac{1}{\alpha}}^\downarrow$ and \eqref{eq:alpha1}.
	%The proof is done by choosing $\alpha$ to reach the supreme value of the exponent $\frac{1-\alpha}{\alpha}\left(\log|\mathcal{Z}|- H^{\downarrow}_{2-\frac{1}{\alpha}}(X{\,|\,}E)\right)$.
\end{proof}

\section{Application: classical-quantum wiretap channel coding}\label{sec:c-q}
We now apply our strong converse Theorem \ref{theo:sc} as well as Dupuis' achievability result \cite{Dup21} to estimate the information leak to the eavesdropper in communication via a classical-quantum (c-q) wiretap channel.
Recall that a c-q wiretap channel $\mathcal{N}_{X\to BE}(\cdot)$ from a classical system $\mathcal{X}$ to the joint quantum system $BE$ is defined as follows,
\begin{align}
	\mathcal{N}_{X\to BE}(\cdot) &:= \sum_{x\in\mathcal{X}} \langle x| \cdot |x\rangle\sigma_{BE}^x\, .
	%\mathcal{N}_{X\to E}(\cdot) &:= \sum_{x\in\mathcal{X}} \langle x| \cdot |x\rangle\sigma_E^x.
\end{align}
If Alice sends a classical symbol $x\in \mathcal{X}$, the channel output states received by Bob and Eve are respectively the marginal states $\sigma_B^x \in \mathcal{S}(\mathcal{H}_B)$ and $\sigma_E^x \in \mathcal{S}(\mathcal{H}_E)$. The goal of Alice is to transmit classical messages from a message set $[M]$ to Bob over a c-q wiretap channel, without leaking too much information to the environment, or simply Eve.
%
%Let the messages be $m\in \{1,\cdots , M\}$, with size $M$. The protocol can be estimated by the cardinality of message $|M|$, average error probability of Bob's received message $\hat{m}$, $\frac{1}{M}\sum_{m=1}^M P(\hat{m}\neq m)$, and the average of $d_1$ norm of Eve's distinguishability of transmitted states.
%

In the protocol, we need two additional ingredients: first, a random codebook subject to $p_X$ on $\mathcal{X}$, where $p_X$ can be arbitrarily chosen \textit{a priori}; second, a strongly $2$-universal hash function $h: [ML] \to [M]$ accessible by both Alice and Bob such that the probability of it being \emph{balanced}\footnote{A hash function $h:[ML] \to [M]$ is balanced if for any $m\in [M]$, $\left|\{k\in[ML]:h(k)=m\}\right| = L$.} is at least $(1-1/ML) $, which we called an \emph{almost balanced strongly $2$-universal hash function}. The existence of latter can be obtained via a standard construction of hash function. The following presents one of such constructions.
\begin{cons}\label{cons}
To map from $[ML] \equiv \{0,1\}^u$ to $[M] \equiv \{0,1\}^v$, we identify $[ML]$ with the Galois field $GF(2^u)$ in the natural way. 
Pick two (uniformly) random numbers $a,b\in GF(2^u)$. For any $x\in [ML]$, define
\begin{align}
h(x) := [ax+b]_{v}
\end{align}
where the calculation $ax+b$ is done over the field $GF(2^u)$ and $[y]_v$ denotes the first $v$ bits of $y\in GF(2^u)\cong \{0,1\}^u$. It is clear that when $a$ is nonzero (hence invertible in a field), 
\[ |\{[ax+b]_{v}=z \}|=|\{ [y]_{v}=z\}|=2^{u-v}=L\pl.\]
\end{cons}
Hence construction \ref{cons} is strongly 2-universal and it has the probability of $(1-\frac{1}{ML})$ being balanced, and the probability of $\frac{1}{ML}$ being a uniform distribution independent of each $x$. In our protocol, the following steps are observed when the hash function is balanced.
\begin{enumerate}[label=\arabic*.]
    \item\label{label:1} Alice uniformly chooses a message $m\in[M]$ to send, i.e.~the state at Alice is
    \begin{align}
        \rho_M = \frac1M\sum_{m\in[M]} |m\rangle \langle m|\label{rhom}\ .
    \end{align}

    \item\label{label:2} Alice picks a random hash function $h$ in the almost balanced strongly $2$-universal family, and applies its reverse function $\mathcal{R}_h^{-1}$ on her message $m\in[M]$ as follows: for every $m\in[M]$, $\mathcal{R}_h^{-1}(m) = k\in{[ML]}$ is uniform for $\{k : h(k)=m \}$. %\textcolor{red}{Li:How this depends on $h$}.\textcolor{blue}{Yu-Chen: I've changed the notation (m,l) to k, I hope this helps}
    Alice's state at this step is
    \begin{align}
        \rho_{ML} = \frac{1}{ML}\sum_{k\in[ML]} \ket{k}\bra{k}\label{rhoml}.
    \end{align}

    \item\label{label:3} For each message $k\in [ML]$, Alice generates the codeword $\mathcal{C}(k) = x_{k} \in \mathcal{X}$ under distribution $p_X$ and announces it publicly.
    Namely, Alice's encoder is described by $\mathcal{E} := \mathcal{C}\circ \mathcal{R}_h^{-1}$.
    Given the codebook $\mathcal{C}$, Alice's state is now
    \begin{align}
        \rho_{MLX}^{\mathcal{C}} = \frac{1}{ML}\sum_{k\in[ML]} \ket{k}\bra{k} \otimes |x_{k}\rangle \langle x_{k}|\label{rhoxml}.
    \end{align}

    \item Alice transmit her codeword $x_{k}$ through the c-q wiretap channel $\mathcal{N}$.
    The joint state between the channel input and output is then
    \begin{align}
        &\sigma^{\mathcal{C}}_{ABE}\equiv\sigma^\mathcal{C}_{MLXBE}
        :=\mathcal{N}^{X\to BE}\left( \rho_{MLX}^{\mathcal{C}} \right) 
        =
        \sum_{ k\in [ML]}\frac{1}{ML}\ket{k}\bra{k}\otimes\ket{x_{k}}\bra{x_{k}}\otimes \sigma_{BE}^{x_{k}}\label{sigmaabe}.
    \end{align}

    \item\label{label:5} Upon receiving the channel output state, Bob performs a positive operator-valued measure (POVM) $\Pi := \left\{\Pi^{x_{k}}_B\right\}_{k\in[ML]}$ to obtain outcome $\hat{x}_{k}$. %\textcolor{red}{Li:Bob should not have assess to $ml$. How the POVM is indexed by ml. What is the definition for $\mathcal{E}^{-1}$.}\textcolor{blue}{I forgot that Alice has to announce the encoding so Bob can decode it. $\mathcal{E}^{-1} = \mathcal{R}^h\circ\mathcal{C}$} 
    He then applies decoding $\mathcal{E}^{-1}$ on $\hat{x}_{k}$ to obtain the estimated message $\hat{m}\in[M]$.
\end{enumerate}

When the hash function is balanced, the average error probability and distinquishability of Eve's state via trace distance under random hash function $h$ and random codebook $\mathcal{C}$ is 
\begin{align}
\epsilon\left(\mathcal{N} \mid \mathcal{E}, \Pi\right) &:=
	\frac{1}{M}\sum_{m\in[M]} \Pr(\hat{m}\neq m\mid \mathcal{E}, \Pi) \\
	&=
	\frac{1}{M}\sum_{m\in[M]} \Tr \left[ \frac{1}{L}\sum_{k\in[L]}\sigma_B^{x_{k}} \left(\mathbb{1}-\sum_{k' \in [L] }\Pi^{x_{k'}}_B\right) \right];\\
		d_1(\mathcal{N}\mid\mathcal{E}) &:=
	\frac{1}{2}\left\| \sigma_{ME}^{\mathcal{C}} - \rho_M\otimes \sigma_{E}^{\mathcal{C}} \right\|_1.
	%\frac{1}{M}\sum_{m\in[M]}\left\|\frac{1}{L}\sum_{l\in[L]}\sigma_E^{x_{m,l}} - \frac{1}{ML}\sum_{(m',l')\in[ML]} \sigma_E^{x_{m',l'}}\right\|_1. \label{cqd1}
\end{align}
For achievability part, if the hash function is unbalanced, we can assume Alice just publicly announce the message. By this assumption, we have the following result.
\begin{theo}[Secrecy exponent for wiretap channel coding] \label{theo:wireach}
	Consider a classical-quantum wiretap channel $\mathcal{N}_{X\to BE}$.
	For any integers L, M, and any prior distribution $p_X$ on $\mathcal{X}$,
	a coding strategy $(\mathcal{E}, \Pi)$ satisfies
	\begin{align}
		\mathbb{E}_{\mathcal{C},h}\left[\epsilon\left(\mathcal{N} \mid \mathcal{E}, \Pi\right)\right] &\leq 4 %(ML)^{\frac{1-\alpha}{\alpha}}
		\e^{\sup_{\frac{1}{2}\leq\alpha\leq 1} \frac{\alpha-1}{\alpha} \left(  I^{\downarrow}_{2-\sfrac{1}{\alpha}}(X;B)_{\sigma} - \log ML \right)},\\
		\mathbb{E}_{\mathcal{C},h} \left[ d_1(\mathcal{N}\mid\mathcal{E})\right] &\leq
		2\e^{\sup_{1<\alpha\leq2} \frac{\alpha-1}{\alpha}(I^{*}_{\alpha}(X{\,:\,}E)_{\sigma} - \log L)},
	\end{align}
where,
	   $\sigma_{XBE} := \sum_{x\in\mathcal{X}} p_X(x) |x\rangle \langle x|\otimes \sigma_{BE}^x$ for each $\sigma_{BE}^{x}$ being the channel output of $\mathcal{N}_{X\to BE}$, and $I^{\downarrow}_{2-\sfrac{1}{\alpha}}$ and $I^{*}_{\alpha}$
	   are defined in \eqref{eq:Petz_Renyi_down_in}.
	   
The secrecy exponent $\sup_{1<\alpha\leq2} \frac{\alpha-1}{\alpha}(I^{*}_{\alpha}(X{\,:\,}E)_{\sigma} - \log L)$ is positive if and only if $\log L > I(X{\,:\,}E)_\sigma$.
\end{theo}

\noindent The proof of Theorem~\ref{theo:wireach} is deferred to Appendix~\ref{sec:proof_secrecy}.

\begin{remark}
An similar result of our Theorem \ref{theo:wireach} has been proved in \cite[Theorem 2]{Hay2112}. There are two main difference between their theorem and ours. First, for \cite[Theorem 2]{Hay2112} the hash function is chosen from a 2-universal family with balanced condition, while in our theorem the hash function is chosen from a strongly 2-universal family which is almost balanced. Second, the distribution $p_X$ used in \cite[Theorem 2]{Hay2112} has to be uniform distribution, while no such limitation exists in our theorem.
\end{remark}

Theorem~\ref{theo:wireach} shows that the secrecy exponent is positive when we use enough randomness in hashing (i.e.~$\log L > I(X{\,:\,}E)_\sigma$).
On the other hand, when the randomness is not enough (i.e.~$\log L < I(X{\,:\,}E)_\sigma$), the following exponential strong converse (i.e.~$d_1 \to 1 $ exponentially fast) can be derived by utilizing our one-shot strong converse bound proved in Theorem~\ref{theo:sc} of Section~\ref{sec:sc}.  
Here, we do not need to specify our protocol for the case that the hash function is unbalanced, because the probability is exponential small. 
\begin{theo}[Exponential strong converse of wiretap channel coding]\label{theo:wiretap}
Let $\mathcal{N}_{X\to BE}$ be a classical-quantum wiretap channel. 
	For any integers $L$ and $M$ and any prior distribution $p_X$ on $\mathcal{X}$, when using the above protocol, the expected distinbuishability of Eve satisfies
	\begin{align}
		\mathbb{E}_{\mathcal{C}, h} \left[ d_1(\mathcal{N}\mid\mathcal{E}) \right] &\geq 1-5\e^{-\sup_{\sfrac{1}{2}<\alpha<1}\frac{1-\alpha}{\alpha}\left(I^{\downarrow}_{2-\sfrac{1}{\alpha}}(X{\,:\,}E)_\sigma - \log L\right)},
	\end{align}
where $\sigma_{XBE} := \sum_{x\in\mathcal{X}} p_X(x) |x\rangle \langle x|\otimes \sigma_{BE}^x$, and $I^{\downarrow}_{2-\sfrac{1}{\alpha}}$ is defined in \eqref{eq:Petz_Renyi_down_in}.
%\textcolor{red}{Li: do we have the downarrow in $I$.}\textcolor{blue}{ Yu-Chen: Yes, thanks for correction}

The  exponent $\sup_{\sfrac{1}{2}<\alpha<1}\frac{1-\alpha}{\alpha}(I^{\downarrow}_{2-\sfrac{1}{\alpha}}(X{\,:\,}E)_\sigma - \log L )$ is positive if and only if $\log L < I(X{\,:\,}E)_\sigma$.
\end{theo}

\noindent The proof of Theorem~\ref{theo:wiretap} is deferred to Appendix~\ref{sec:proof_esc}.

\medskip

Those two bounds established in Theorems~\ref{theo:wireach} and \ref{theo:wiretap} can be applied similarly when Alice communicate classical information through a quantum channel instead of a c-q wiretap channel.
Consider a quantum channel $\mathcal{N}_{A'\rightarrow B}$ from Alice to Bob. Note that $A'$ is now a quantum system.
Given an isometric extension $U^{\mathcal{N}}_{A'\rightarrow BE}$ of the channel $\mathcal{N}_{A'\rightarrow B}$ via Stinespring dilation, the complementary channel $\mathcal{N}_{A'\rightarrow E}$ to the environment or Eve, is
\begin{align}
	\mathcal{N}_{A'\rightarrow E}(\cdot) = \Tr_B\left[U^{\mathcal{N}}_{A'\rightarrow BE}(\cdot)\right].
\end{align}
Given a c-q coding scheme $x\mapsto \rho_{A'}^x$, the protocol when the hash function is balanced can be structured similarly with the protocol of c-q wiretap channel as follows.
\begin{enumerate}[label=\arabic*.]
\item[1.--3.] The first three steps are exactly the same as Step~\ref{label:1}, \ref{label:2}, and \ref{label:3} of the c-q channel version 
%Alice uniformly chooses a message $m\in [M]$ to send. The state at Alice is \eqref{rhom}.
by the following substitutions:
\begin{align}
\begin{cases}
	\mathcal{N}_{X\to B} \leftarrow \mathcal{N}_{A'\to B}\\%\label{exchange1}
	\sigma_{B}^{x} \leftarrow \mathcal{N}_{A'\to B}\left(\rho^{x}_{A'}\right)\\%\label{exchange2}
	\end{cases}.
\end{align}
%The expression $\eqref{eq:XBE}$ can be exchanged similarly
%\begin{align}
%	&\rho_{XA'}:= \sum_x p_X(x)\ket{x}\bra{x}\otimes\rho^{x}_{A'}\label{exchange3}\\
%	&\sigma_{XBE}:= \mathcal{N}_{A'\to BE}(\rho_{XA'})\label{exchange4}
%\end{align}

%2. Alice randomly chooses a hash function from the family of hash function, $\{f_h: h\in \mathscr{H}\}$ and applies its reverse function $f_h^{-1}$ on her message $m\in[M]$ as follows: For every $m\in[M]$, $f_h^{-1}(m) = (m,l)$ is uniform for all $l\in [L]$. Alice's state at this step is \eqref{rhoml}.\\
%3. \\
\setcounter{enumi}{3}
\item Alice encodes the $x_{k}$ to a quantum state $\rho^{x_{k}}_{A'} \in \mathcal{S}(\mathcal{H}_{A'})$. Hence, the encoded state is
\begin{align}
     \rho_{MLA'}^{\mathcal{C}} = \frac{1}{ML}\sum_{k\in[ML]} |k\rangle \langle k| \otimes \rho^{x_{k}}_{A'} \,.
\end{align}
\item Alice sends the state $\rho^{x_{k}}_{A'}$ through the quantum channel $\mathcal{N}_{A'\to BE}$ to have state
\begin{align}
        \sigma^\mathcal{C}_{MLBE} &:= U^\mathcal{N}_{A'\to BE}\left(\rho_{MLA'}^{\mathcal{C}} \right)\\
        &=
        \sum_{k\in [ML]}\frac{1}{ML}\ket{k}\bra{k}\otimes U^\mathcal{N}_{A'\to BE} \left(\rho_{A'}^{x_{k}}\right).
\end{align}
\item Similar to Step~\ref{label:5} of the c-q version.
\end{enumerate}
Following similar analysis in the case of c-q wiretap channel, we get the following results.
\begin{coro}[Secrecy exponent for private communication over quantum channels] \label{coro:private_ach}
Let ${U}^{\mathcal{N}}_{A'\rightarrow BE}(\cdot)$
be a Stinespring dilation of a quantum channel $\mathcal{N}_{A'\to B}$.
	For every integers $L$ and $M$, any prior distribution $p_X$ on $\mathcal{X}$ and mapping $x\mapsto \rho_{A'}^x$, when using the above protocol,
	a coding strategy $(\mathcal{E}, \Pi)$ satisfies
	\begin{align}
		\mathbb{E}_{\mathcal{C},h} \left[\epsilon\left(\mathcal{N} \mid \mathcal{E}, \Pi\right) \right] &\leq 4 %(ML)^{\frac{1-\alpha}{\alpha}}
		\e^{\sup_{\frac{1}{2}\leq\alpha\leq 1} \frac{\alpha-1}{\alpha} \left(  I^{\downarrow}_{2-\sfrac{1}{\alpha}}(X;B)_{\sigma} - \log ML \right)},\\
		\mathbb{E}_{\mathcal{C},h} \left[ d_1(\mathcal{N}\mid\mathcal{E}) \right] &\leq
		2\e^{\sup_{1<\alpha\leq2} \frac{\alpha-1}{\alpha}(I^{*}_{\alpha}(X{\,:\,}E)_{\sigma} - \log L)},
	\end{align}
where $\sigma_{XBE} := \sum_{x\in\mathcal{X}} p_X(x) |x\rangle \langle x|\otimes U^\mathcal{N}_{A'\to BE}(\rho_{A'}^x)$.
\end{coro}

\begin{coro}[Exponential strong converse for private communication over quantum channels] \label{coro:private_sc}
Let ${U}^{\mathcal{N}}_{A'\rightarrow BE}(\cdot)$
be a Stinespring dilation of a quantum channel $\mathcal{N}_{A'\to B}$.
	For any integers $L$ and $M$, any prior distribution $p_X$ on $\mathcal{X}$ and mapping $x\mapsto \rho_{A'}^x$, when using the above protocol, the expected distinbuishability of Eve is bounded by
	\begin{align}
		\mathbb{E}_{\mathcal{C}, h} \left[ d_1(\mathcal{N}\mid\mathcal{E})\right] &\geq 1-5\e^{-\sup_{\sfrac{1}{2}<\alpha<1}\left(\frac{1-\alpha}{\alpha}(I^{\downarrow}_{2-\frac{1}{\alpha}}(X{\,:\,}E)_\sigma - \log L)\right)}\, ,
	\end{align}
where $\sigma_{XBE} := \sum_{x\in\mathcal{X}} p_X(x) |x\rangle \langle x|\otimes U^\mathcal{N}_{A'\to BE}(\rho_{A'}^x)$.
\end{coro}

\section{Application: entropy accumulation} \label{sec:QKD}
The main topic of entropy accumulation (EA) \cite{DFR20,Dup21} is to measure how much uncertainty remains about the bitstring $A^n_1$ given access to side information $X^n_1$. In fact, the EA protocol can be generally generated as the figure below:
\begin{figure}[H]
\vspace{-0.8em}
\includegraphics[width = 0.5\columnwidth]{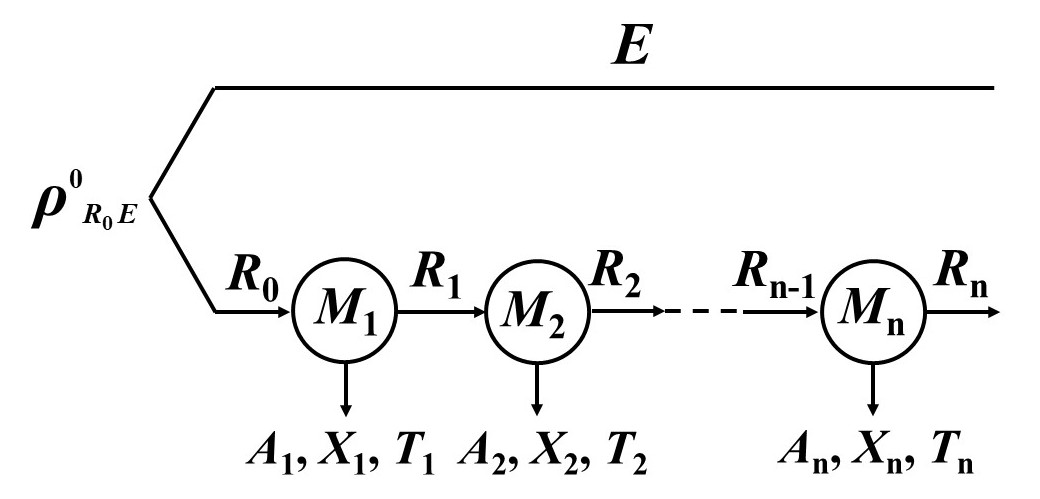}
\end{figure}
Each $M_i$ is a process that passes information on to the next one using a ``memory'' register $R_i$.
Note that the additional messages $T^n_1$ are global information about $A_1^n$ and $X_1^n$. For instance, in cryptographic scenarios, each process $M_i$ can be inferred by tests carried out by quantum key distribution protocols on some generated bits. Then, $T_i$ should tell us whether position $i$ is tested, and if so, the outcome of the test performed at step $i$. Here we restrict $A_1^n$ and $X_1^n$ to classical systems and $T^n_1$ being classical bits.

In \cite[Theorem 9]{Dup21}, an upper bound of information about the bitstring is shown: for hash functions $h:\mathcal{A}\to \mathcal{Z}$ with $\log|\mathcal{Z}|=nR$, if $0<f(w)-R\leq \frac{V^2}{2}$, then
\begin{align}
	\frac{1}{2}\mathbb{E}_h\left\|\mathcal{R}^h(\rho_{A^n_1X^n_1E|\wt(T^n_1)=w})-\frac{\mathds{1}}{2^{nR}}\otimes\rho_{X^n_1E|\wt(T^n_1)=w}\right\|_1
	\leq  \frac{1}{\Pr[\wt(T_1^n)=w]}\e^{-\frac{n}{2}\left(\frac{R-f(w)}{V}\right)^2}.
\end{align}
where $\wt(T_1^n)$ is the hamming weight of $T_1^n$,
and $f(\cdot)$ is called a \emph{tradeoff function} (defined below) telling us how much entropy is expected to get given the probability of seeing $1$ on $T_i$, $V$ is a constant larger than 2. 
% The strict definition of it is stated below. 
\begin{defn*}[Definition 4.1 in \cite{DFR20}]
%Let $\mathbb{P}$ be the set of probability distributions on the alphabet $\mathcal{T}$ of $T_i$, for any $w \in \mathbb{P}$ we define the set of states
%\begin{align}
%\Sigma_i(w) =\left\{\nu_{T_iA_iX_iR_iE} = (M_i\otimes \mathds{1}_E)(\omega_{R_{i-1}E}): \omega \in \mathds{S}(R_{i-1}\otimes E) \text{ and } \nu_{T_i} = w\right\}.
%\end{align}
A real function $f$ on probability distributions is called a \emph{min-tradeoff function} for process $M_i$ if 
\begin{align}
f(w) \leq \inf_{\nu\in \Sigma_i(w)} H(A_i{\,|\,}X_i E)_{\nu},
\end{align}
and a real function $f$ on probability distributions is called a \emph{max-tradeoff function} for process $M_i$ if
\begin{align}
f(w) \geq \sup_{\nu\in \Sigma_i(w)} H(A_i{\,|\,}X_i E)_{\nu}.
\end{align}
where $\Sigma_i(w)$ is the set of states $\nu_{T_iA_iX_iR_iE}$ with probability of $T_i = 1$ being $w$.
\end{defn*}

Using Theorem \ref{theo:sc}, we provide an exponential strong converse bound complementary to the above result \cite[Theorem 9]{Dup21}. 
\begin{theo}\label{theo:QKD}
	For any strongly $2$-universal hash function $ h: \mathcal{A} \to \mathcal{Z}$ with $\log  |\mathcal{Z}|=nR$, if $0<R-f(w)<V$, then
	\begin{align}
		\frac{1}{2}\mathbb{E}_h\left\|\mathcal{R}^h(\rho_{A^n_1X^n_1E|\wt(T^n_1)=w})-\frac{\mathds{1}}{2^{nR}}\otimes\rho_{X^n_1E|\wt(T^n_1)=w}\right\|_1
		\geq 1-\frac{4}{\Pr[\wt(T_1^n)=w]}\e^{-\frac{n}{2}\left(\frac{R-f(w)}{V}\right)^2}.
	\end{align}
\end{theo}
% Above applies when the rate is close to the tradeoff function $f(w)$. 
% \begin{remark}
% The tradeoff function $f(\cdot)$ in the strong converse version (which is max-tradeoff function in \cite[Definition 4.1]{DFR20}) is not the same of  $f(\cdot)$ in the direct version (which is min-tradeoff function in \cite[Definition 4.1]{DFR20}).
% \end{remark}
\begin{proof}[Proof of Theorem \ref{theo:QKD}]
	Using Theorem \ref{theo:sc},
	\begin{align}
		&\frac{1}{2}\mathbb{E}_h\left\|\mathcal{R}^h(\rho_{A^n_1X^n_1E|\wt(T^n_1)=w})--\frac{\mathds{1}}{2^{nR}}\otimes\rho_{X^n_1E|\wt(T^n_1)=w}\right\|_1\\
		\geq& 1-4\e^{-\frac{1-\alpha}{\alpha}\left(nR-H^{\downarrow}_{2-\frac{1}{\alpha}}(A^n_1|X^n_1E|\wt(T^n_1)=w)\right)}\\
		\overset{(a)}{\geq}& 1-4\e^{-\frac{1-\alpha}{\alpha}\left(nR-H^{\downarrow*}_{2-\frac{1}{\alpha}}(A^n_1|X^n_1E|\wt(T^n_1)=w)\right)}\\
		\overset{(b)}{=} &1-4\e^{-\frac{\alpha'-1}{\alpha'}\left(nR-H^{\downarrow*}_{\frac{1}{\alpha'}}(A^n_1|X^n_1E|\wt(T^n_1)=w)\right)}\\
		\overset{(c)}{\geq}& 1-4\e^{-\frac{\alpha'-1}{\alpha'}\left(nR-nf(w)-n\left(\frac{\alpha'-1}{4}V^2\right) - \frac{\alpha'}{\alpha'-1}\log\frac{1}{\Pr[\wt(T_1^n)=w]}\right)}\\
		=& 1-\frac{4}{\Pr[\wt(T_1^n)=w]}\e^{-\frac{\alpha'-1}{\alpha'}\left(nR-nf(w)-n\left(\frac{\alpha'-1}{4}V^2\right)\right)}
	\end{align}
	In (a), $H_{\alpha}^{\downarrow*}(X{\,|\,}E)_{\rho} := \frac{1}{1-\alpha}\log \tr \left[  (\rho_E^{\frac{1-\alpha}{\alpha}}\rho_{XE}\rho_E^{\sfrac{(1-\alpha)}{\alpha}})^{\alpha} \right]$. In (b), $\frac{1}{\alpha'} = 2-\frac{1}{\alpha}$. Since $\frac{1}{2}<\alpha<1$, $\alpha'\geq1$. In (c) we use \cite[Proposition 4.5]{DFR20}: for $1<\alpha'<1+\frac{2}{V} < 2$ the following bound holds:
	\begin{align}
		&H^{\downarrow*}_{\frac{1}{\alpha'}}(A^n_1|X^n_1E|\wt(T^n_1)=w)\leq nf(w) + n\left(\frac{\alpha'-1}{4}\right)V^2 + \frac{\alpha'}{\alpha'-1}\log\frac{1}{\Pr[\wt(T_1^n)=w]}.
	\end{align}
	Let $\beta := \frac{\alpha'-1}{\alpha'}$, the exponent become
	\begin{align}
		\frac{\alpha'-1}{\alpha'}\left(nR-nf(w)-n\left(\frac{\alpha'-1}{4}V^2\right)\right)
		&\geq \frac{\alpha'-1}{\alpha'}\left(nR-nf(w)-n\left(\frac{\alpha'-1}{2\alpha'}V^2\right)\right)\\
		& = n\left(\beta(R-f(w))-\beta^2(\frac{V^2}{2})\right)
	\end{align}
		Since $1<\alpha'<1+\sfrac2V$, then $0<\beta<\frac{2}{V+2}$. Under the assumption $0<R-f(w)<V< \frac{2V^2}{V+2}$,
 the exponent can be maximized to $n(\frac{1}{2}\frac{(R-f(w))^2}{V^2})$ at $\beta = \frac{R-f(w)}{V^2}<\frac{2}{V+2}$. That completes the proof.
	%Hence we get
	%\begin{align}
	%&\frac{1}{2}\mathbb{E}_h\left\|\mathcal{R}^h(\rho_{A^n_1X^n_1E|\wt(T^n_1)=w})--\frac{\mathds{1}}{2^{nR}}\otimes\rho_{X^n_1E|\wt(T^n_1)=w}\right\|_1\\
	%&\geq 1-\frac{4}{\Pr[\wt(T_1^n)=w]}\e^{-n\frac{1}{2}\left(\frac{R-f(w)}{V}\right)^2}
	%\end{align}
	%which is equivalent to our statement.
\end{proof}
%\begin{align}
%&\Pr[\wt(T^n_1) = w]\cdot \left(1-\mathbb{E}_h \left.\frac{1}{2}\right\|\mathcal{R}^h(\rho_{A^n_1X^n_1E|\wt(T^n_1)=w})\right.\\
%&\quad\left.\left.-\frac{\mathds{1}}{2^{nR}}\otimes\rho_{X^n_1E|\wt(T^n_1)=w}\right\|_1 \right) \leq 4\e^{-n \left(\frac{1}{2}\frac{(R-f(w))^2}{V^2}\right)}.
%\end{align}

\section{Moderate Deviation Analysis} \label{sec:moderate}
We shall now applies our
 one-shot strong converse to the \emph{moderate deviation regime} \cite{CH17, CTT2017}.
We call $(a_n)_{n\in \mathds{N}}$ a positive moderate deviation sequence if it satisfies
\begin{align}
        \lim_{n\to \infty }a_n=0 \ ,\   
          \lim_{n\to \infty }n a_n^2=\infty. \label{eq:an}
\end{align}
In other words, the moderate deviation sequence vanishes (i.e.~of order $o(1)$) but slower than $O(\sfrac{1}{\sqrt{n}})$ (i.e.~of order $\omega(\sfrac{1}{\sqrt{n}})$).

For each $n\in\mathds{N}$, let us denote $R_n := \frac1n \log |\mathcal{Z}^n|$ as the rate of the extracted randomness. 
%  Let $|\mathcal{Z}^n| = \e^{nR_n}$ with $R_n = H(X{\,|\,}E)_\rho - a_n$. 
We have the following moderate deviation characterization for privacy amplification against quantum side information when the rate $R_n$ approaches the first-order limit $H(X{\,|\,}E)_\rho$ at speed of $a_n$.
\begin{prop}[Moderate deviations for privacy amplification]\label{prop:moderate_PA}
	Let $\rho_{XE}$ be a classical-quantum state and assume that $V(X{\,|\,}E)_\rho>0$. Suppose $(a_n)_{n\in \mathds{N}}$ is a moderate deviation sequence. Then for any strongly $2$-universal hash function $h^n: \mathcal{X}^n\to \mathcal{Z}^n$,
	\begin{align}
	\begin{dcases}
	\liminf_{n\to \infty} - \frac{1}{n a_n^2} \log \left(  \frac12\mathds{E}_{h} \left\|  \left(\mathcal{R}^{h^n} - \mathcal{U}^{n} \right) \left(\rho_{XE}^{\otimes n}\right) \right\|_1 \right) \geq \frac{1}{2 V(X{\,|\,}E)_\rho},
	& \text{ if  } |\mathcal{Z}^n| = \e^{n(H(X{\,|\,}E)_\rho - a_n)} \\
	\liminf_{n\to \infty} - \frac{1}{n a_n^2} \log \left( 1 - \frac12\mathds{E}_{h} \left\|  \left(\mathcal{R}^{h^n} - \mathcal{U}^{n} \right) \left(\rho_{XE}^{\otimes n}\right) \right\|_1 \right) \geq \frac{1}{2 V(X{\,|\,}E)_\rho}, &  \text{ if  } |\mathcal{Z}^n| = \e^{n(H(X{\,|\,}E)_\rho + a_n)}
	\end{dcases}.
	\end{align}
where $\mathcal{U}^n$ is the perfectly randomizing channel from $\mathcal{X}^n$ to $\mathcal{Z}^n$.
\end{prop}
The above theorem means that the trace distance vanishes asymptotically when the rate $R_n$ approaches $H(X{\,|\,}E)_\rho$ from below at speed no faster than $O(\sfrac{1}{\sqrt{n}})$.
On the other hand, in the strong converse regime where $R_n > H(X{\,|\,}E)_\rho$, the trace distance still converges to $1$ asymptotically when the rate of maximal extractable randomness approaches $H(X{\,|\,}E)_\rho$ from above at speed no faster than $O(\sfrac{1}{\sqrt{n}})$.

Before proving Theorem \ref{prop:moderate_PA}, we state the lemmas about derivative of entropic quantities that will be used in the proof.
\begin{lemm}[{\cite[Proposition 11]{Tomhay16}}]\label{lemm:updiff}
For every classical-quantum state $\rho_{XE}$, $\alpha \mapsto H^*_{\alpha}(X{\,|\,}E)_{\rho}$ and $\alpha \mapsto I^*_{\alpha}(X{\,:\,}E)_{\rho}$ are continuously differentiable on $\alpha \in [1,2]$. Moreover,
\begin{align}
\left.\frac{\mathrm{d}}{\mathrm{d}\alpha}H^*_{\alpha}(X{\,|\,}E)_{\rho}\right|_{\alpha = 1} = -\frac{V(X{\,|\,}E)_{\rho}}{2}\, ,\ \ \ 
\left.\frac{\mathrm{d}}{\mathrm{d}\alpha}I^*_{\alpha}(X{\,:\,}E)_{\rho}\right|_{\alpha = 1} = \frac{V(X{\,:\,}E)_{\rho}}{2}\,.
\end{align}
\end{lemm}

\begin{lemm}[\cite{CHD+21},\cite{CH17}]\label{lemm:downdiff}
For every classical-quantum state $\rho_{XE}$, $\alpha \mapsto H^{\downarrow}_{2-\sfrac{1}{\alpha}}(X{\,|\,}E)_{\rho}$ and $\alpha \mapsto I^{\downarrow}_{2-\sfrac{1}{\alpha}}(X{\,:\,}E)_{\rho}$ are analytical on $\alpha \in \left[\sfrac{1}{2},1\right]$. Moreover,
\begin{align}
\left.\frac{\mathrm{d}}{\mathrm{d}\alpha}H^{\downarrow}_{2-\sfrac{1}{\alpha}}(X{\,|\,}E)_{\rho}\right|_{\alpha = 1} = -\frac{V(X{\,|\,}E)_{\rho}}{2}\, ,\ \ \ 
\left.\frac{\mathrm{d}}{\mathrm{d}\alpha}I^{\downarrow}_{2-\sfrac{1}{\alpha}}(X{\,:\,}E)_{\rho}\right|_{\alpha = 1} = \frac{V(X{\,:\,}E)_{\rho}}{2}\,.
\end{align}
\end{lemm}
\begin{proof}[Proof of Theorem \ref{prop:moderate_PA}]
We start with the first claim. 
Let $R_n = H(X{\,|\,}E)_{\rho} - a_n$.
By \cite[Theorem 8]{Dup21}, one has, for every $n\in\mathds{N}$,
\begin{align}
\frac12\mathds{E}_{h} \left\|  \left(\mathcal{R}^{h^n} - \mathcal{U}^{n} \right) \left(\rho_{XE}^{\otimes n}\right) \right\|_1 \leq  \mathrm{e}^{-n \sup_{\alpha\in(1,2]}\frac{1-\alpha}{\alpha}  \left( R_n -  H_{\alpha}^{*}(X{\,|\,}E)_\rho\right) }.
\end{align}
Using Lemma \ref{lemm:updiff}, we can apply Taylor's series expansion of $\alpha\mapsto H^*_{\alpha}(X{\,|\,}E)_{\rho}$ at $\alpha=1$:
\begin{align}
H^*_{\alpha}(X{\,|\,}E)_{\rho} = H(X{\,|\,}E)_{\rho} - (\alpha-1)\frac{V(X{\,|\,}E)_{\rho}}{2} + \mathscr{R}(\alpha-1),
\end{align}
where $\mathscr{R}(\alpha-1)$ is a continuous function satisfying $\frac{\mathscr{R}(\alpha-1)}{\alpha-1}\to 0$ as $\alpha\to 1$.
Using the above expansion, the fact that $R_n = H(X{\,|\,}E)_{\rho} - a_n$, and letting  $\alpha_n =  1+\frac{a_n}{V(X{\,|\,}E)_\rho}$ for sufficient large $n\in\mathds{N}$ such that $\alpha_n \in (1,2]$, we have
\begin{align}
\sup_{\alpha\in (1,2]} \left\{\frac{1-\alpha}{\alpha}  \left( R_n -  H_{\alpha}^{*}(X{\,|\,}E)_\rho\right)\right\} &\geq \frac{1-\alpha_n}{\alpha_n}  \left( R_n -  H_{\alpha_n}^{*}(X{\,|\,}E)_\rho\right)\\
&= \frac{1}{1+\frac{a_n}{V(X{\,|\,}E)_{\rho}}}\left(\frac{a_n^2}{2V(X{\,|\,}E)_{\rho}} + \frac{a_n^2}{V(X{\,|\,}E)_{\rho}^2}\frac{\mathscr{R}(\alpha_n-1)}{\alpha_n-1}\right)\\
&= \frac{a_n^2}{2V(X{\,|\,}E)_{\rho}}\frac{1}{1+\frac{a_n}{V(X{\,|\,}E)_{\rho}}}\left(1+\frac{2}{V(X{\,|\,}E)_{\rho}}\frac{\mathscr{R}(\alpha_n-1)}{\alpha_n-1}\right).%\\
%& \geq \frac{a^2_n}{2V(X{\,|\,}E)_\rho} - \frac{a_n^3}{2V(X{\,|\,}E)_{\rho}^3}\Upsilon.
\end{align}
In other words,
\begin{align}
- \frac{1}{n a_n^2} \log \left(  \frac12\mathds{E}_{h} \left\|  \left(\mathcal{R}^{h^n} - \mathcal{U}^{n} \right) \left(\rho_{XE}^{\otimes n}\right) \right\|_1 \right) \geq  \frac{1}{2V(X{\,|\,}E)_{\rho}}\frac{1}{1+\frac{a_n}{V(X{\,|\,}E)_{\rho}}}\left(1+\frac{2}{V(X{\,|\,}E)_{\rho}}\frac{\mathscr{R}(\alpha_n-1)}{\alpha_n-1}\right).
\end{align}
Recalling that $\displaystyle \lim_{n\to\infty} a_n = 0$ and $\displaystyle\frac{\mathscr{R}(\alpha_n-1)}{\alpha_n-1}\to0$,
we obtain the lower bound as desired, i.e.~
\begin{align}
\liminf_{n\to \infty} - \frac{1}{n a_n^2} \log \left(  \frac12\mathds{E}_{h} \left\|  \left(\mathcal{R}^{h^n} - \mathcal{U}^{n} \right) \left(\rho_{XE}^{\otimes n}\right) \right\|_1 \right) \geq \frac{1}{2 V(X{\,|\,}E)_\rho}.
\end{align}
This proves the first claim. For the second claim, we have by Theorem \ref{theo:sc} that for every $n\in\mathds{N}$,
\begin{align}
 1 - \frac12\mathds{E}_{h} \left\|  \left(\mathcal{R}^{h^n} - \mathcal{U}^{n} \right) \left(\rho_{XE}^{\otimes n}\right) \right\|_1 &\leq 4\, \mathrm{e}^{-n\sup_{\alpha\in(\sfrac12, 1)}
	\frac{1-\alpha}{\alpha} \big( R_n -  H_{2-\sfrac{1}{\alpha}}^{\downarrow}(X{\,|\,}E)_\rho\big)}.
\end{align}
Using Lemma \ref{lemm:downdiff}, we apply Taylor's series expansion again on $\alpha \mapsto H^{\downarrow}_{2-\sfrac{1}{\alpha}}(X{\,|\,}E)_{\rho}$ at $\alpha=1$:
\begin{align}
H^{\downarrow}_{2-\sfrac{1}{\alpha}}(X{\,|\,}E)_{\rho} = H(X{\,|\,}E)_{\rho} - (\alpha-1)\frac{V(X{\,|\,}E)_{\rho}}{2} + \frac{(\alpha-1)^2}{2} \left.\frac{\mathrm{d}^2}{\mathrm{d} \alpha^2}H^{\downarrow}_{2-\frac{1}{\alpha}}(X{\,|\,}E)_{\rho}\right|_{\alpha = \bar{\alpha}}
\end{align}
for some $\bar{\alpha}\in[\alpha,1]$. Using the above expansion, the fact that $R_n = H(X{\,|\,}E)_{\rho} + a_n$, and let $\alpha_n =  1-\frac{a_n}{V(X{\,|\,}E)_\rho}$, for all $\sfrac{1}{2}<\alpha<1$:
\begin{align}
&\sup_{\sfrac{1}{2}<\alpha<1}\left\{
	\frac{1-\alpha}{\alpha} \big( R_n -  H_{2-\sfrac{1}{\alpha}}^{\downarrow}(X{\,|\,}E)_\rho\big)\right\}  \\ \geq & \frac{1-\alpha_n}{\alpha_n} \big( R_n -  H_{2-\sfrac{1}{\alpha_n}}^{\downarrow}(X{\,|\,}E)_\rho\big)\\
=& \frac{1}{1-\frac{a_n}{V(X{\,|\,}E)_{\rho}}}\left(\frac{a_n^2}{2V(X{\,|\,}E)_{\rho}} - \frac{a_n^3}{2V(X{\,|\,}E)_{\rho}^3}\left.\frac{d^2}{d \alpha^2}H^{\downarrow}_{2-\frac{1}{\alpha}}(X{\,|\,}E)_{\rho}\right|_{\alpha = \bar{\alpha}_n}\right)\\
\geq & \frac{1}{1-\frac{a_n}{V(X{\,|\,}E)_{\rho}}}\left(\frac{a_n^2}{2V(X{\,|\,}E)_{\rho}} - \frac{a_n^3}{2V(X{\,|\,}E)_{\rho}^3}\Upsilon\right),%\\
%& \geq \frac{a^2_n}{2V(X{\,|\,}E)_\rho} - \frac{a_n^3}{2V(X{\,|\,}E)_{\rho}^3}\Upsilon.
\end{align}
where 
\begin{align}
\Upsilon = \max\limits_{\sfrac{1}{2}\leq\alpha\leq 1}\left|\frac{d^2}{d \alpha^2}H^{\downarrow}_{2-\frac{1}{\alpha}}(X{\,|\,}E)_{\rho}\right|
\end{align}
is finite due to the extreme value theorem together with the closed set $[\sfrac{1}{2},1]$ and $\frac{d^2}{d \alpha^2}H^{\downarrow}_{2-\frac{1}{\alpha}}(X{\,|\,}E)_{\rho}$ is continuous for $\alpha\in [\sfrac{1}{2},1]$ as stated in Lemma \ref{lemm:downdiff}. Hence,
\begin{align}
- \frac{1}{n a_n^2} \log \left( 1 - \frac12\mathds{E}_{h} \left\|  \left(\mathcal{R}^{h^n} - \mathcal{U}^n \right) \left(\rho_{XE}^{\otimes n}\right) \right\|_1 \right) \geq -\frac{\log 4}{na_n^2} + \frac{1}{2V(X{\,|\,}E)_{\rho}}\frac{1}{1-\frac{a_n}{V(X{\,|\,}E)_{\rho}}}\left(1-\frac{\alpha_n\Upsilon}{V(X{\,|\,}E)^2_{\rho}}\right).
\end{align}
Letting $n\to \infty$ and using the definition of $a_n$,
\begin{align}
	\liminf_{n\to \infty} - \frac{1}{n a_n^2} \log \left( 1 - \frac12\mathds{E}_{h} \left\|  \left(\mathcal{R}^{h^n} - \mathcal{U}^n \right) \left(\rho_{XE}^{\otimes n}\right) \right\|_1 \right) \geq \frac{1}{2 V(X{\,|\,}E)_\rho}\,,
\end{align}
which is our second claim.
\end{proof}

In the following Theorem~\ref{prop:moderate_wiretap}, we establish moderate deviation analysis on the information leakage when communication through a classical-quantum wiretap channel studied in Section~\ref{sec:c-q}.
Namely, we establish the asymptotic behaviors of the security criterion $d_1$ when the rate of used randomness, i.e.~$\frac1n \log |L^n|$ approaches $I(X{\,:\,}E)_\sigma$ at the speed no faster than $O(\sfrac{1}{\sqrt{n}})$.

\begin{prop}[Moderate deviations for classical-quantum wiretap channel] \label{prop:moderate_wiretap}
	Consider an arbitrary classical-quantum wiretap channel $\mathcal{N}_{X\to BE}$, any prior distribution $p_X$ on $\mathcal{X}$ satisfying $V(X{\,:\,}E)_{\sigma}>0$. Using the protocol introduced in Section \ref{sec:c-q}, we have the following result for any moderate deviation sequence $(a_n)_{n\in\mathds{N}}$ defined in \eqref{eq:an}:
	\begin{align}
	\begin{dcases}
	\liminf_{n\to \infty} - \frac{1}{n a_n^2} \log \left(  \mathbb{E}_{\mathcal{C}^n,h^n}d_1(\mathcal{N}^{\otimes n}\mid\mathcal{E}^n)\right) \geq \frac{1}{2 V(X{\,:\,}E)_\sigma},
	& \text{ if } |L^n| = \e^{n(I(X{\,:\,}E)_\sigma + a_n)} \\
	\liminf_{n\to \infty} - \frac{1}{n a_n^2} \log \left( 1-\mathbb{E}_{\mathcal{C}^n,h^n}d_1(\mathcal{N}^{\otimes n}\mid\mathcal{E}^n) \right) \geq \frac{1}{2 V(X{\,:\,}E)_\sigma}, & \text{ if } |L^n| = \e^{n(I(X{\,:\,}E)_\sigma - a_n)}
	\end{dcases}.
	\end{align}
	Here,  $\sigma_{XBE} := \sum_{x\in\mathcal{X}} p_X(x) |x\rangle \langle x|\otimes \sigma_{BE}^x$ for each $\sigma_{BE}^{x}$ being the channel output of $\mathcal{N}_{X\to BE}$.
\end{prop}
This means that the information leakage decays to $0$ asymptotically when the number of bits of randomness in hashing approaches $I(X{\,:\,}E)_{\sigma}$ from above at speed no faster than $O(\sfrac{1}{\sqrt{n}})$, and it converges to $1$ asymptotically when the number of bits approaches $I(X{\,:\,}E)_{\sigma}$ from below.
\begin{proof}
For the first claim, by Theorem \ref{theo:wireach},
\begin{align}
\mathbb{E}_{\mathcal{C}^n,h^n}d_1(\mathcal{N}^{\otimes n}\mid\mathcal{E}n) &\leq
		2\e^{-n\sup\limits_{1<\alpha\leq2}\frac{1-\alpha}{\alpha}\left(I^{*}_{\alpha}(X{\,:\,}E)_{\sigma} - R_n\right)}.
\end{align}
where $R_n=\frac{1}{n}\log |L^n|$ is the rate of the used randomness in hashing.
Using Lemma \ref{lemm:updiff}, we can apply Taylor theorem of $I^*_{\alpha}(X{\,:\,}E)_{\rho}$ at $\alpha=1$. 
\begin{align}
I^*_{\alpha}(X{\,:\,}E)_{\sigma} = I(X{\,:\,}E)_{\sigma} + \frac{\alpha-1}{2}V(X{\,:\,}E)_{\sigma} + \mathscr{R}(\alpha -1),
\end{align}
where $\mathscr{R}(\alpha-1)$ is a continuous function satisfying $\frac{\mathscr{R}(\alpha-1)}{\alpha-1}\to 0$ as $\alpha\to 1$. Let $\alpha_n =  1+\frac{a_n}{V(X{\,:\,}E)_{\sigma}}$. Using the above expansion and $R_n = I(X{\,:\,}E)_{\sigma} + a_n$, we have $1<\alpha_n\leq2$ for all sufficiently large $n\in\mathds{N}$, and
\begin{align}
\max\limits_{1<\alpha\leq2}\left\{\frac{1-\alpha}{\alpha}(I^{*}_{\alpha}(X{\,:\,}E)_{\sigma} - R_n)\right\}
&\geq \frac{1-\alpha_n}{\alpha_n}(I^{*}_{\alpha_n}(X{\,:\,}E)_{\sigma} - R_n)\\
&= \frac{1}{1+\frac{a_n}{V(X{\,:\,}E)_{\sigma}}}\left(\frac{a_n^2}{2V(X{\,:\,}E)_{\sigma}}-\frac{a_n^2}{V(X{\,:\,}E)_{\sigma}^2}\frac{\mathscr{R}(\alpha_n-1)}{\alpha_n-1}\right)\\
& =\frac{a_n^2}{2V(X{\,:\,}E)_{\sigma}}\frac{1}{1+\frac{a_n}{V(X{\,:\,}E)_{\sigma}}}\left(1-\frac{2}{V(X{\,:\,}E)_{\sigma}}\frac{\mathscr{R}(\alpha_n-1)}{\alpha_n-1}\right).
%&\geq \frac{a^2_n}{2V(X{\,:\,}E)_{\sigma}} - \frac{a_n^3}{2V(X{\,:\,}E)_{\sigma}^3}\Upsilon
\end{align}
Hence,
\begin{align}
- \frac{1}{n a_n^2} \log \left(\mathbb{E}_{\mathcal{C}^n,h^n}d_1(\mathcal{N}^{\otimes n}\mid\mathcal{E}^n) \right) \geq -\frac{\log 2}{na_n^2}+\frac{1}{2V(X{\,:\,}E)_{\sigma}}\frac{1}{1+\frac{a_n}{V(X{\,:\,}E)_{\sigma}}}\left(1-\frac{2}{V(X{\,:\,}E)_{\sigma}}\frac{\mathscr{R}(\alpha_n-1)}{\alpha_n-1}\right).
\end{align}
Taking $n\to \infty$ and using the definition of $a_n$,
\begin{align}
	\liminf_{n\to \infty} - \frac{1}{n a_n^2} \log \left(  \mathds{E}_{\mathcal{C}^n,h^n} d_1(\mathcal{N}^{\otimes n}\mid\mathcal{E}^n)\right) \geq \frac{1}{2 V(X{\,:\,}E)_\sigma},
\end{align}
which proves our first claim. 

For the second claim, note that by Theorem \ref{theo:wiretap}
\begin{align}
1-\mathbb{E}_{\mathcal{C}^n,h^n}d_1(\mathcal{N}^{\otimes n}\mid\mathcal{E}^n) &\leq 5\e^{-n\sup_{\sfrac{1}{2}<\alpha<1}\frac{1-\alpha}{\alpha}\left(I^{\downarrow}_{2-\sfrac{1}{\alpha}}(X{\,:\,}E)_\sigma - R_n\right)} .
\end{align}
Using Lemma \ref{lemm:downdiff}, we apply Taylor's series expansion of $\alpha \mapsto I^{\downarrow}_{2-\sfrac{1}{\alpha}}(X{\,:\,}E)_\sigma$ at $\alpha=1$:
\begin{align}
I^{\downarrow}_{2-\sfrac{1}{\alpha}}(X{\,:\,}E)_\sigma = I(X{\,:\,}E)_{\sigma} + \frac{(\alpha-1)}{2}V(X{\,:\,}E)_{\sigma} +\frac{(\alpha-1)^2}{2}\left.\frac{\mathrm{d}^2}{\mathrm{d}s^2}\left(I^{\downarrow}_{2-\frac{1}{\alpha}}(X{\,:\,}E)_\sigma\right)\right|_{\alpha = \bar{\alpha}} 
\end{align}
for some $\bar{\alpha}\in [\alpha,1]$. Let $\alpha_n = 1-\frac{a_n}{V(X{\,:\,}E)_{\sigma}}$. Using the above expansion and the fact that $R_n = I(X{\,:\,}E)_{\sigma} - a_n$, we have for all $ \sfrac{1}{2}<\alpha_n<1$,
\begin{align}
&\sup_{\sfrac{1}{2}<\alpha<1}\left\{\sfrac{1-\alpha}{\alpha}\left(I^{\downarrow}_{2-\frac{1}{\alpha}}(X{\,:\,}E)_\sigma - R_n\right)\right\} \\ \geq & \frac{1-\alpha_n}{\alpha_n}(I^{\downarrow}_{2-\sfrac{1}{\alpha_n}}\left(X{\,:\,}E\right)_\sigma - R_n)\\
=& \frac{1}{1-\frac{a_n}{V(X{\,:\,}E)_{\sigma}}}\left(\frac{a_n^2}{2V(X{\,:\,}E)_{\sigma}} + \frac{a_n^3}{2V(X{\,:\,}E)_{\sigma}^3}\left.\frac{\mathrm{d}^2}{\mathrm{d}\alpha^2}I^{\downarrow}_{2-\frac{1}{\alpha}}(X{\,:\,}E)_\sigma\right|_{a = \bar{a}_n}\right)\\
\geq& \frac{1}{1-\frac{a_n}{V(X{\,:\,}E)_{\sigma}}}\left(\frac{a_n^2}{2V(X{\,:\,}E)_{\sigma}} - \frac{a_n^3}{2V(X{\,:\,}E)_{\sigma}^3}\Upsilon\right),
\end{align}
where $\Upsilon = \max\limits_{\alpha\in[\sfrac{1}{2},1]}\left|\frac{\mathrm{d}^2}{\mathrm{d}\alpha^2}I^{\downarrow}_{2-\frac{1}{\alpha}}(X{\,:\,}E)_{\sigma}\right|$; and this quantity is finite due to $[\sfrac{1}{2},1]$ being closed,  $\frac{\mathrm{d}^2}{\mathrm{d}\alpha^2}I^{\downarrow}_{2-\frac{1}{\alpha}}(X{\,:\,}E)_{\sigma}$ is continuous for $\alpha\in [\sfrac{1}{2},1]$, as stated in Lemma \ref{lemm:downdiff}, and the extreme value theorem. Hence, 
\begin{align}
- \frac{1}{n a_n^2} \log \left(1-\mathbb{E}_{\mathcal{C}^n,h^n}d_1(\mathcal{N}^{\otimes n}\mid\mathcal{E}^n) \right) \geq -\frac{\log 5}{na_n^2}+\frac{1}{2V(X{\,:\,}E)_{\sigma}}\frac{1}{1-\frac{a_n}{V(X{\,:\,}E)_{\sigma}}}\left(1-\Upsilon\frac{a_n}{V(X{\,:\,}E)_{\sigma}^2}\right).
\end{align}
Taking $n\to \infty$ and using the definition of $a_n$, we obtain
\begin{align}
\liminf_{n\to \infty} - \frac{1}{n a_n^2} \log \left( 1-\mathbb{E}_{\mathcal{C}^n,h^n}d_1(\mathcal{N}^{\otimes n}\mid\mathcal{E}^n) \right) \geq \frac{1}{2 V(X{\,:\,}E)_\sigma}\ ,
\end{align}
which proves our second claim.
\end{proof}

Finally, the moderate deviation can also be applied in entropy accumulation. We omit the proof since it is similar to previous two moderate deviation analysis.
\begin{prop}[Moderate deviations for entropy accumulation] \label{prop:moderate_EA}
	Consider any strongly 2-universal hash function: $h^n:\mathcal{A}^n\to \mathcal{Z}^n$ that produce $\log|\mathcal{Z}^n| = nR$ at the output. In the protocol of entropy accumulation, the following holds:
	\begin{align}
	\begin{dcases}
	\liminf_{n\to \infty} - \frac{1}{n a_n^2} \log \left(  \frac{1}{2}\mathbb{E}_h\left\|\mathcal{R}^h(\rho_{A^n_1X^n_1E|\wt(T^n_1)=w})-\frac{\mathds{1}}{2^{nR}}\otimes\rho_{X^n_1E|\wt(T^n_1)=w}\right\|_1\right) \geq \frac{1}{2V^2},
	&\!\!\!\! \text{ if } R = f(w)-a_n \\
	\liminf_{n\to \infty} - \frac{1}{n a_n^2} \log \left( 1-\frac{1}{2}\mathbb{E}_h\left\|\mathcal{R}^h(\rho_{A^n_1X^n_1E|\wt(T^n_1)=w})-\frac{\mathds{1}}{2^{nR}}\otimes\rho_{X^n_1E|\wt(T^n_1)=w}\right\|_1 \right) \geq \frac{1}{2V^2}, &\!\!\!\! \text{ if } R = f(w)+a_n
	\end{dcases}.
	\end{align}
	where $f(\cdot)$ for the first equation is the min-tradeoff function defined in \cite[Definition 4.1]{DFR20}, and $f(\cdot)$ for the second equation is the max-tradeoff function defined in \cite[Definition 4.1]{DFR20}. 
\end{prop}
The above result shows that by Dupuis' result \cite[Theorem 9]{Dup21} the information we get at the output becomes almost uncertain when the number of bits at the output approaches $f(w)$ from below at speed no faster than $O(\sfrac{1}{\sqrt{n}})$, and it becomes almost certain when the number of bits approaches $f(w)$ from above.

\section{Conclusions} \label{sec:conclusion}
We establish a one-shot strong converse bound for privacy amplification against quantum side information, which enjoys various advantages as the recent achievability bound by Dupuis~\cite[Theorem 8]{Dup21}. Moreover, our result extends to the large deviation regime \cite{Hay07, CH16, CHT19, CHDH2-2018, Hao-Chung, CHD+21}---an exponential convergence to $1$ for every blocklength, and the moderate deviation regime \cite{CH17, CTT2017}---an asymptotic behavior of trace distance when the rate of the extracted randomness approaches the quantum conditional entropy.
In a way, our result in the strong converse regime may be viewed as complementing Dupuis' result \cite[Theorem 8]{Dup21} in the error exponent regime.
As an application, we provide both secrecy exponent bound and exponential strong converse bound for the information leakage through a classical-quantum wiretap channel as well as for a quantum channel. Our result also applies to estimate the information loss in entropy accumulation protocol \cite{DFR20, Dup21}, and those two applications can be also extended to the large deviation regime and moderate deviation regime characterizations.

We remark that several entropic quantities such as $H_\alpha^*$ and $I_\alpha^*$ do not have closed-form expressions for $\alpha\neq 1$ in general. There is a recent optimization algorithm with asymptotic convergence guarantee that can be applied to compute them \cite{YCL21}.
It is intriguing to note that some entropic exponent functions obtained in this paper such as $H_\alpha^{\downarrow}$ and $I_\alpha^{\downarrow}$ have the same form as classical-quantum channel coding \cite{Hay07, CHDH2-2018} and classical data compression with quantum side information \cite{CHD+21}.

\section*{Acknowledgement}
H.-C.~Cheng would like to thank Kai-Min Chung for his insightful discussions, and also thank Masahito Hayashi for his comments on some of our early results.
Y.-C.~Shen and H.-C.~Cheng are supported by the Young Scholar Fellowship (Einstein Program) of the Ministry of Science and Technology in Taiwan (R.O.C.) under Grant MOST 110-2636-E-002-009, and are supported by the Yushan Young Scholar Program of the Ministry of Education in Taiwan (R.O.C.) under Grant NTU-110V0904,  Grant NTU-CC-111L894605, and Grand NTU-111L3401.

\appendix
\section{Auxiliary Proofs} \label{sec:proofs}

\subsection{Proof of a trace inequality} \label{sec:proof_trace}
\begin{proof}[Proof of Lemma~\ref{lemm:trace}]
% 	To prove our claim, we invoke a technique in quantum state discrimination.
% 	When discriminating density operators $\rho$ and $\sigma $ with equal prior, Barnum and Knill proved that the error probability of discrimination is upper bounded by $2$ times the error probability of using the optimal measurements \cite{BK02}, \cite[Theorem 3.10]{Wat18}, i.e.~ an inequality as follows,
% 	\begin{align}
% 		\Tr\left[ \rho (\rho + \sigma)^{-\frac12} \sigma (\rho + \sigma)^{-\frac12} \right] 
% 		\leq  \Tr\left[ \frac{  \rho + \sigma - \left| \rho -  \sigma \right|}{2} \right].
% 	\end{align}
% 	We claim that this inequality can be extended to positive semi-definite operators $K,L \geq0$, i.e.
    We first claim the following, for all positive semi-definite operators $K,L$,
	\begin{align} \label{eq:BK}
		\Tr\left[ K (K+L)^{-\frac12} L (K+L)^{-\frac12} \right] 
		\leq \Tr\left[ \frac{K+L - |K-L|}{2} \right].
	\end{align}
	Then, combining it with Audenaert \textit{et al.}'s inequality \cite[Theorem 2]{ANS+08}: for all $K,L \geq 0$,
	\begin{align}
		\Tr\left[ \frac{K+L-|K-L|}{2} \right] \leq \Tr\left[ K^{1-s} L^s \right], \quad \forall s\in (0,1),
	\end{align}
	we  prove Lemma \ref{lemm:trace}. To prove \eqref{eq:BK}, we adapt Barnum and Knill's proof technique in \cite{BK02} and \cite[Theorem 3.10]{Wat18}. Let $\rho_0 = K\geq 0$, $\rho_1 = L \geq 0$, $M:=\rho_0+\rho_1$, and let $\Pi_0=\{\rho_0\ge \rho_1\}, \Pi_1=\{\rho_0<\rho_1\}=1-\Pi_0$ be the corresponding optimal measurement \cite{Hel67, Hol72, AM14}. Then
	\begin{align}
		\Tr\left[ \frac{M + |\rho_0 -\rho_1|}{2} \right]
		=\Tr[\rho_0\Pi_0] + \Tr[\rho_1\Pi_1]
		\label{eq:max}
	\end{align}
	Using the Cauchy--Schwarz inequality for trace, we have for each $i=1,2$,
	\begin{align}
		\Tr[\rho_i\Pi_i] &\le  \left\|M^{-\frac14}\rho_i M^{-\frac14}\right\|_2 \left\| M^{\frac14} \Pi_i M^{\frac14} \right\|_2. 
	\end{align}
	Applying the Cauchy--Schwarz inequality for scalars, we have
	\begin{align}
		\sum_{i\in\{0,1\}} \Tr[\rho_i\Pi_i]
		\leq \sqrt{\sum_{i\in\{0,1\}}\left\|M^{-\frac14}\rho_i M^{-\frac14}\right\|_2^2} \sqrt{\sum_{i\in\{0,1\}} \left\|  M^{\frac14} \Pi_i M^{\frac14} \right\|_2^2}. \label{eq:CauchyShwarz}
	\end{align}
	The second factor on the right-hand side is bounded as follows:
	\begin{align}
		&\sum_{i\in\{0,1\}} \left\|  M^{\frac14} \Pi_i M^{\frac14} \right\|_2^2
		= \sum_{i\in\{0,1\}} \Tr\left[  \Pi_i \sqrt{M}\Pi_i \sqrt{M}\right] 
		\leq \sum_{i\in\{0,1\}} \Tr\left[ \sqrt{M}\Pi_i \sqrt{M}\right] 
		= \Tr[M]. \label{eq:second}
	\end{align}
	On the other hand, the first factor is
	\begin{align}
		 \sum_{i\in\{0,1\}}\left\|M^{-\frac14}\rho_i M^{-\frac14}\right\|_2^2
		= \sum_{i\in\{0,1\}} \Tr\left[ \rho_i M^{-\frac12} \rho_i M^{-\frac12} \right] 
		= \Tr[M] - 2\Tr\left[ \rho_0 M^{-\frac12} \rho_1 M^{-\frac12} \right]. \label{eq:first}
	\end{align}
	Combining \eqref{eq:max}, \eqref{eq:CauchyShwarz}, \eqref{eq:second}, and \eqref{eq:first} together, we obtain
	\begin{align}
		2 \frac{\Tr\left[ \rho_0 M^{-\frac{1}{2}} \rho_1 M^{-\frac{1}{2}} \right]}{\Tr[M]}
		&\leq 1 - \left( 1 - \frac{ \Tr\left[ \frac{M - |\rho_0 - \rho_1|}{2} \right] }{\Tr[M]} \right)^2 \\
		&= \frac{ \Tr\left[ \frac{M - |\rho_0 - \rho_1|}{2} \right] }{\Tr[M]} \left( 2 - \frac{ \Tr\left[ \frac{M - |\rho_0 - \rho_1|}{2} \right] }{\Tr[M]} \right) \\
		&\leq 2\frac{ \Tr\left[ \frac{M - |\rho_0 - \rho_1|}{2} \right] }{\Tr[M]}, \label{eq:end}
	\end{align}
	where in the last inequality we have used 
	\begin{align}
		\Tr\left[ \frac{M  -  |\rho_0 - \rho_1|}{2} \right]
		= \Tr\left[ \rho_0 - \left( \rho_0 - \rho_1\right)_+ \right] \geq 0.
		%= \Tr\left[ \rho_0\left\{ \rho_0\leq \rho_1 \right\} + \rho_1\left\{ \rho_0> \rho_1 \right\} \right] \geq 0.
	\end{align}
	Then, \eqref{eq:end} is exactly our claim \eqref{eq:BK}, and hence we complete the proof. 
\end{proof}

\subsection{Proof of secrecy exponent for wiretap channel coding} \label{sec:proof_secrecy}

\begin{proof}[Proof of Theorem~\ref{theo:wireach}]
The expected value (over the random codebook) of the average error probability of the protocol was stated in \cite[Equation (63)]{Hay132}:
\begin{align}
	\mathbb{E}_\mathcal{C}\left[\epsilon\left(\mathcal{N} \mid \mathcal{E},\Pi \right)\right] \leq \min_{\frac12 \leq \alpha\leq 1} 4 (ML)^{\frac{1-\alpha}{\alpha}} \e^{\frac{\alpha-1}{\alpha} I^{\downarrow}_{2-\sfrac{1}{\alpha}}(X;B)_{\sigma}},
	\label{error}
\end{align}
We remain to prove the upper bound on $d_1$.
It is sufficient to consider the event that the output of hash function is balanced. In this situation, $\mathcal{E}$ and $\mathcal{E}^{-1}$ can be expressed as:
\begin{align}
	&\mathcal{E}: \ket{m}\bra{m} \mapsto \frac{1}{L}\sum_{k: \mathcal{R}^h(k) = m}\ket{k}\bra{k}\otimes \ket{x_{k}}\bra{x_{k}}\,;\\
	&\mathcal{E}^{-1}: \sigma \mapsto \sum_{k\in[ML]} (\bra{k}\otimes\bra{x_{k}})\sigma (\ket{k}\otimes\ket{x_{k}}) \otimes \ket{h(x)}\bra{h(x)}.
\end{align}

The $d_1$ norm of the channel when the hash function is balanced can be expressed as
\begin{align}
	d_1(\mathcal{N}\mid\mathcal{E}) 
	&= 	\frac{1}{2}\frac{1}{M}\sum_{m\in[M]}\left\|\frac{1}{L}\sum_{k: \mathcal{R}^h(k) = m}\sigma_E^{x_{k}} - \frac{1}{ML}\sum_{k'\in[ML]} \sigma_E^{x_{k'}}\right\|_1\\
	%\frac{1}{M}\sum_{m=1}^M\left\|\frac{1}{L}\sum_{l}^L\sigma_E^{x_{m,l}} - \frac{1}{ML}\sum_{m',l'}^{M,L} \sigma_E^{x_{m',l'}}\right\|_1\\
		&= \frac{1}{2}\left\| \frac{1}{M} \sum_{m\in[M]} \ket{m}\bra{m}\otimes\left(\frac{1}{L}\sum_{k\in h^{-1}(m)}\sigma_E^{x_{k}}\right)
		-\frac{1}{M} \sum_{m\in[M]} \ket{m}\bra{m}\otimes \left(\frac{1}{ML}\sum_{k' \in [ML]}\sigma_E^{x_{k'}}\right)\right\|_1 \\
		&= \frac{1}{2}\left\| \left(\mathcal{E}^{-1} - \mathcal{U}^{\mathcal{C}}\right)\left(\sigma^{\mathcal{C}}_{MLXE}\right) \right\|_1\\
	&= \frac{1}{2}\left\| \left(\mathcal{E}^{-1} - \mathcal{U}^{\mathcal{C}}\right)\left(\sigma^{\mathcal{C}}_{AE}\right) \right\|_1 \, ,
\end{align}
where $\mathcal{U}^{\mathcal{C}}$ is the perfectly randomizing channel from $\ket{k}\bra{k}\otimes\ket{x_{k}}\bra{x_{k}}$ to $\ket{m}\bra{m}$ for $k \in [ML]$ and $m\in [M]$, and $A\equiv MLX$ is a classical system. On the other hand, the probability of $h$ being not balanced is at most $\frac{1}{ML}$, and the $d_1$ norm in that case is at most 1. 

Recall that any strongly $2$-universal family of hash functions is $1$-randomizing (e.g.~\cite[Lemma 6]{Dup21}). Hence, given a codebook $\mathcal{C}: [ML] \rightarrow \mathcal{X}$, we now apply \cite[Theorem 8]{Dup21} or namely \eqref{dupuis} to obtain an upper bound to the expected value of Eve's distinguishability with respect to the family of hash functions,
\begin{align}
	\mathbb{E}_{h\mid\mathcal{C}} \left[ d_1(\mathcal{N}_{X\to E}\mid \mathcal{E}) \right]
	&\leq\frac{1}{2}\mathbb{E}_{h\mid\mathcal{C}} \left\| (\mathcal{E}^{-1} - \mathcal{U}^{\mathcal{C}})(\sigma^\mathcal{C}_{AE}) \right\|_1 + \frac{1}{ML}\\
	&\leq \e^{\frac{\alpha-1}{\alpha}\left(\log M - H_{\alpha}^*(A{\,|\,}E)_{\sigma^{\mathcal{C}}}\right)} + \frac{1}{ML}, \quad \forall \alpha \in [1,2].
	%&\leq 2\cdot\e^{\frac{\alpha-1}{\alpha}\left(\log M - H_{\alpha}(A{\,|\,}E)_{\sigma^{\mathcal{C}}_{AE}}\right)}
\end{align}
Since in balanced condition, Alice sends every $k\in [ML]$ to the codebook $\mathcal{C}$ with equal probability, we calculate that
\begin{align}
	I_{\alpha}^{*}(A;E)_{\sigma^{\mathcal{C}}_{AE}}
	&=\inf_{ \tau_E \in \mathcal{S}(E)} D_{\alpha}^{*}\left(\sigma^{\mathcal{C}}_{AE}\|\sigma^{\mathcal{C}}_A\otimes \tau_E\right)\\
	&= \inf_{ \tau_E \in \mathcal{S}(E)} D_{\alpha}^{*}\left(\bigoplus\limits_{k}\frac{1}{ML} \sigma_E^{x_{k}} \left\| \frac{\mathds{1}_A}{ML}\otimes\tau_E\right.\right)\\
	%&= \min\limits_{ \sigma_E } \frac{1}{\alpha - 1} \log \Tr\left[[(\frac{\mathbb{1}}{ML}\otimes\sigma_E)^{\frac{1-\alpha}{2\alpha}} (\bigoplus\limits_{m,l}\frac{1}{ML}(\mathcal{N}_{X\to E}(\ket{x_{m,l}}\bra{x_{m,l}}))(\frac{\mathbb{1}}{ML}\otimes\sigma_E)^{\frac{1-\alpha}{2\alpha}}]^{\alpha}\right]\\
	&= \inf_{ \tau_E \in \mathcal{S}(E)} D_{\alpha}^{*}\left(\left.\bigoplus\limits_{k}\frac{1}{ML} \sigma_E^{x_{k}} \right\| \mathds{1}_A\otimes\tau_E \right) + \log(ML)\\
	&= -H_{\alpha}^*(A{\,|\,}E)_{\sigma^{\mathcal{C}}_{AE}} + \log(ML). 
\end{align}
Hence, expectation of Eve's distinguishability using random codebook $\mathcal{C}$ can further be written as
\begin{align}
	\mathbb{E}_{h\mid\mathcal{C}} \left[ d_1(\mathcal{N}_{X\to E}|\mathcal{E}) \right]
	&\leq  \e^{\frac{\alpha-1}{\alpha}\left(\log M - H_\alpha(A{\,|\,}E)_{\sigma^{\mathcal{C}}_{AE}}\right)} + \frac{1}{ML}\\
	&= \e^{\frac{\alpha-1}{\alpha}\left(\log M + I_{\alpha}^{*}(A;E)_{\sigma^{\mathcal{C}}_{AE}} - \log ML\right)} + \frac{1}{ML}\\
	&= \e^{\frac{\alpha-1}{\alpha}\left(I_{\alpha}^{*}(A;E)_{\sigma^{\mathcal{C}}_{AE}} - \log L\right)} + \frac{1}{ML}\\
	&\leq 2\cdot\e^{\frac{\alpha-1}{\alpha}\left(I_{\alpha}^{*}(A;E)_{\sigma^{\mathcal{C}}_{AE}} - \log L\right)}.
\end{align}
The last inequality comes from the fact that the term $\sfrac{1}{ML}$ decays faster than the first term, i.e.~$\log ML \geq \frac{\alpha-1}{\alpha}(\log L -I_{\alpha}^{*}(A;E)_{\sigma^{\mathcal{C}}_{AE}}) $ for every $\alpha\in(1,2]$ and noting that
$I_{\alpha}^{*}(A;E)_{\sigma^{\mathcal{C}}_{AE}}\geq 0$.
Next, invoking the concavity of the map $\sigma^{\mathcal{C}}_{A} \mapsto 2\cdot\e^{\frac{\alpha-1}{\alpha}\left(I_{\alpha}^{*}(A;E)_{\sigma^{\mathcal{C}}_{AE}} - \log L\right)}$ which is proved in Lemma \ref{lemm:concavity}, expectation can be taken over the random codebook $\mathcal{C}$ (under distribution $p_X$) by applying Jensen's inequality:
\begin{align}
	\mathbb{E}_{\mathcal{C},h} \left[ d_1(\mathcal{N} \mid \mathcal{E}) \right]
	&=\mathbb{E}_{\mathcal{C}}\mathbb{E}_{h\mid\mathcal{C}} \left[ d_1(\mathcal{N}|\mathcal{E})\right] \\
	&\leq 2\, \mathbb{E}_\mathcal{C} \left[ \e^{\frac{\alpha-1}{\alpha}\left(I_{\alpha}^{*}(A;E)_{\sigma^{\mathcal{C}}_{AE}} - \log L\right)} \right]\\
	&\leq 2\e^{\frac{\alpha-1}{\alpha}\left(I_{\alpha}^{*}(A;E)_{\mathbb{E}_\mathcal{C}\left[\sigma^{\mathcal{C}}_{AE}\right]} - \log L\right)}.
\end{align}
Moreover, we have
\begin{align}
	\mathbb{E}_\mathcal{C}\left[\sigma^\mathcal{C}_{AE}\right]
	&= \sum_{k\in[ML]}\frac{1}{ML}\ket{k}\bra{k}\otimes\left(\sum_{x\in\mathcal{X}}p_X(x)\ket{x}\bra{x}\otimes \sigma_E^{x}\right).
\end{align}
By simple calculation, one has
\begin{align}
	I_{\alpha}^{*}(A{\,:\,}E)_{\mathbb{E}_\mathcal{C}\left[\sigma^{\mathcal{C}}_{AE}\right]} = I_{\alpha}^{*}(X{\,:\,}E)_{\sigma},
\end{align}
Therefore,
\begin{align}
	\mathbb{E}_{\mathcal{C},h} \left[ d_1(\mathcal{N}\mid\mathcal{E}) \right]
	\leq
	2\e^{\frac{\alpha-1}{\alpha}(I_{\alpha}^{*}(X{\,:\,}E)_{\sigma} - \log L)},\label{d1}
\end{align}
which is our statement.

Finally, the positivity of the secrecy exponent follows from the monotone increasing of the map $\alpha\mapsto I^{*}_{\alpha}(X{\,:\,}E)_{\sigma}$ and \eqref{eq:alpha1}.

\begin{lemm}[A concavity property]\label{lemm:concavity}
	For every $\alpha > 1$, the map
	\begin{align}
		\sigma_X \mapsto \mathrm{e}^{ \frac{\alpha-1}{\alpha}I_{\alpha}^{*}(X{\,:\,}E)_{\sigma}}
	\end{align}
	is concave on all probability distributions on $\mathcal{X}$.
\end{lemm}
\noindent The proof of Lemma~\ref{lemm:concavity} is deferred to Appendix~\ref{sec:proof_concavity}.
\end{proof}

\subsection{Proof of exponential strong converse for wiretap channel coding} \label{sec:proof_esc}

\begin{proof}[Proof of Theorem~\ref{theo:wiretap}]
Same as the beginning of the proof in Theorem \ref{theo:wireach},
the $d_1$ norm of the channel when the hash function is balanced can be expressed as
\begin{align}
	d_1(\mathcal{N}\mid\mathcal{E}) 
	\overset{(a)}= \frac{1}{2}\left\| \left(\mathcal{E}^{-1} - \mathcal{U}^{\mathcal{C}}\right)\left(\sigma^{\mathcal{C}}_{AE}\right) \right\|_1, 
\end{align}
where $\mathcal{U}^{\mathcal{C}}$ is the perfectly randomizing channel from $\ket{k}\bra{k}\otimes\ket{x_{k}}\bra{x_{k}}$ to $\ket{m}\bra{m}$ for $k \in [ML]$ and $m\in [M]$, and $A\equiv MLX$ is a classical system. Since the probability of $h$ being unbalanced is at most $\frac{1}{ML}$ and the value of $d_1$ norm is at most 1, we have by applying Theorem \ref{theo:sc},
\begin{align}
\mathbb{E}_{h\mid\mathcal{C}}\left[ d_1(\mathcal{N}\mid\mathcal{E}) \right]
&\geq\mathbb{E}_{h|\mathcal{C}} \frac{1}{2}\left\| \left(\mathcal{E}^{-1} - \mathcal{U}^{\mathcal{C}}\right)\left(\sigma^{\mathcal{C}}_{AE}\right) \right\|_1 - \frac{1}{ML}\\
&\geq 1-4\cdot\e^{-\frac{1-\alpha}{\alpha}\left(\log M-H^{\downarrow}_{2-\sfrac{1}{\alpha}}(A{\,|\,}E)_{\sigma^{\mathcal{C}}_{AE}}\right)} - \frac{1}{ML}, \quad \forall \alpha \in \left(\sfrac12, 1\right),
\end{align}
Since in balanced condition, Alice sends every $k\in [ML]$ to the codebook $\mathcal{C}$ with equal probability, we calculate that
\begin{align}
	I^{\downarrow}_{2-\frac{1}{\alpha}}(A;E)_{\sigma^{\mathcal{C}}_{AE}}\
	&= D_{2-\frac{1}{\alpha}}\left(\sigma^{\mathcal{C}}_{AE}\|\sigma^{\mathcal{C}}_A\otimes \sigma^{\mathcal{C}}_E\right)\\
	&=  D_{2-\frac{1}{\alpha}}\left(\bigoplus\limits_{k}\frac{1}{ML} \sigma_E^{x_{k}} \left\| \frac{\mathds{1}_A}{ML}\otimes\sigma^{\mathcal{C}}_E\right.\right)\\
	%&= \min\limits_{ \sigma_E } \frac{1}{\alpha - 1} \log \Tr\left[[(\frac{\mathds{1}}{ML}\otimes\sigma_E)^{\frac{1-\alpha}{2\alpha}} (\bigoplus\limits_{m,l}\frac{1}{ML}(\mathcal{N}_{X\to E}(\ket{x_{m,l}}\bra{x_{m,l}}))(\frac{\mathds{1}}{ML}\otimes\sigma_E)^{\frac{1-\alpha}{2\alpha}}]^{\alpha}\right]\\
	&=  D_{2-\frac{1}{\alpha}}\left(\left.\bigoplus\limits_{k}\frac{1}{ML} \sigma_E^{x_{k}} \right\| \mathds{1}_A\otimes\sigma^{\mathcal{C}}_E \right) + \log(ML)\\
	&= -H^{\downarrow}_{2-\frac{1}{\alpha}}(A{\,|\,}E)_{\sigma^{\mathcal{C}}_{AE}} + \log(ML).
\end{align}
Using above relation and taking expectation of the codeword, we obtain
\begin{align}
\mathbb{E}_{h,\mathcal{C}} \left[ d_1(\mathcal{N}\mid\mathcal{E}) \right]
&=\mathbb{E}_{\mathcal{C}}\mathbb{E}_{h|\mathcal{C}} \left[ d_1(\mathcal{N}\mid\mathcal{E}) \right]\\
&\geq \mathbb{E}_{\mathcal{C}}\left[1-4\cdot\e^{-\frac{1-\alpha}{\alpha}\left(I^{\downarrow}_{2-\frac{1}{\alpha}}(A{\,:\,}E)_{\sigma^{\mathcal{C}}_{AE}}-\log L\right)}\right] -\frac{1}{ML}\\
&\overset{(a)}{\geq} \mathbb{E}_{\mathcal{C}}\left[1-5\cdot\e^{-\frac{1-\alpha}{\alpha}\left(I^{\downarrow}_{2-\sfrac{1}{\alpha}}(A{\,:\,}E)_{\mathbb{E}_{\mathcal{C}}\left(\sigma^{\mathcal{C}}_{AE}\right)}-\log L\right)}\right]  \\
&\overset{(b)}{\geq} 1-5\cdot\e^{-\frac{1-\alpha}{\alpha}\left(I^{\downarrow}_{2-\sfrac{1}{\alpha}}(A{\,:\,}E)_{\mathbb{E}_{\mathcal{C}}\left[\sigma^{\mathcal{C}}_{AE}\right]}-\log L\right)} \, ,
\end{align}
where (a) is because the term $\sfrac{1}{ML}$ decays faster than the second term, i.e.~for every $\alpha \in (\sfrac12,1)$,
\begin{align}
    I^{\downarrow}_{2-\frac{1}{\alpha}}(A{\,:\,}E)_{\sigma^{\mathcal{C}}_{AE}} &\leq I(A{\,:\,}E)_{\sigma^{\mathcal{C}}_{AE}} \leq \log(ML);
\end{align}
% comes from the fact that $I^{\downarrow}_{2-\frac{1}{\alpha}}(A{\,:\,}E)_{\sigma^{\mathcal{C}}_{AE}}\leq \log(ML)$\footnote{For $  \alpha\in [\sfrac{1}{2}, \infty], I^{\downarrow}_{2-\frac{1}{\alpha}}(A{\,:\,}E)_{\sigma^{\mathcal{C}}_{AE}} = -H^{\downarrow}_{2-\frac{1}{\alpha}}(A{\,|\,}E)_{\sigma^{\mathcal{C}}_{AE}} + \log(ML) \leq -H_{\alpha}(A{\,|\,}E)_{\sigma^{\mathcal{C}}_{AE}} + \log(ML) = I_{\alpha}(A{\,:\,}E)_{\sigma^{\mathcal{C}}_{AE}}\leq \log(ML)$ where the first inequality uses \cite[Corollary 5.3]{Marco} and the last inequality utilizes \cite[Proposition 3]{CHE20}. Here, $H_{\alpha}(X{\,|\,}E)_{\rho} := -\inf\limits_{\sigma_{E}\in \mathcal{S}(\mathcal{H}_E)}D_{\alpha}(\rho_{XE}||\mathds{1}_X\otimes\rho_E)$ and $I_{\alpha}(X{\,:\,}E)_{\rho}:= \inf \limits_{\sigma_{E}\in \mathcal{S}(\mathcal{H}_E)} D_{\alpha}(\rho_{XE}||\rho_X\otimes\sigma_E)$.}
 and the last inequality (b) utilizes the fact that the map
\begin{align} 
\sigma^{\mathcal{C}}_{A} \mapsto \e^{-\frac{1-\alpha}{\alpha}\left(I^{\downarrow}_{2-\frac{1}{\alpha}}(A{\,:\,}E)_{\sigma^{\mathcal{C}}_{AE}}-\log L\right)}
\end{align} is linear by directly inspecting the definition given in \eqref{eq:Petz_Renyi_down_in}.

Now, note that
\begin{align}
	\mathbb{E}_\mathcal{C}\left[\sigma^\mathcal{C}_{AE}\right]
	&= \sum_{k\in[ML]}\frac{1}{ML}\ket{k}\bra{k}
	\otimes\left(\sum_{x\in\mathcal{X}}p_X(x)\ket{x}\bra{x}\otimes \sigma_E^{x}\right).
\end{align}
By simple calculation, we have
\begin{align}
	I_{\alpha}^{\downarrow}(A;E)_{\mathbb{E}_\mathcal{C}\left[\sigma^{\mathcal{C}}_{AE}\right]} = I_{\alpha}^{\downarrow}(X{\,:\,}E)_{\sigma},
\end{align}
where $\sigma_{XE}$ is given in Theorem \ref{theo:wiretap}.
Putting them all together and choosing $\alpha$ to maximize the exponent yields our result:
\begin{align}
\mathbb{E}_{\mathcal{C}, h} \left[d_1(\mathcal{N}\mid\mathcal{E})\right] &\geq 1-5\e^{-\sup_{\sfrac{1}{2}<\alpha<1}\frac{1-\alpha}{\alpha}\left(I^{\downarrow}_{2-\sfrac{1}{\alpha}}(X{\,:\,}E)_\sigma - \log L\right)}\ .
\end{align}

The positivity of the  exponent follows from the monotone increasing of the map $\alpha\mapsto I^{\downarrow}_{1 - \sfrac{1}{\alpha} }(X{\,:\,}E)_{\sigma}$ and \eqref{eq:alpha1}.
\end{proof}

\subsection{Proof of a concavity property} \label{sec:proof_concavity}
% \begin{lemm}[A concavity property]\label{lemm:concavity}
% 	For every $\alpha > 1$, the map
% 	\begin{align}
% 		\sigma_X \mapsto \mathrm{e}^{ \frac{\alpha-1}{\alpha}I_{\alpha}^{*}(X{\,:\,}E)_{\sigma}}
% 	\end{align}
% 	is concave on all probability distributions on $\mathcal{X}$.
% \end{lemm}
\begin{proof}[Proof of Lemma~\ref{lemm:concavity}]
	As stated in the proof of Proposition 4-(b) in \cite{CGH18}, for any classical-quantum state $\sigma_{XE} = \sum_{x\in\mathcal{X}} p_X(x)|x\rangle \langle x|\otimes \sigma_E^x$ and an arbitrary state $\tau_E \in \mathcal{D}(E)$, we have 
	\begin{align}
		D_{\alpha}^*(\sigma_{XE}\|\sigma_X\otimes \tau_E) = \frac{1}{\alpha-1}\log \sum\limits_{x}p_X(x) \mathrm{e}^{(\alpha-1)D^*_{\alpha}(\rho_E^x\|\tau_E)}.
	\end{align}
	Hence, 
	\begin{align}
		\mathrm{e}^{\frac{\alpha-1}{\alpha}D_{\alpha}^*(\sigma_{XE}\|\sigma_X\otimes \tau_E)} = \left(\sum\limits_{x}p_X(x)\mathrm{e}^{(\alpha-1)D^*_{\alpha}(\rho_E^x\|\tau_E)}\right)^{\frac{1}{\alpha}}.
	\end{align}
	Since the power function $(\,\cdot\,)^{\sfrac{1}{\alpha}}$ is concave for $\alpha>1$, the map $\sigma_X\mapsto \e^{\frac{\alpha-1}{\alpha}D_{\alpha}^*(\sigma_{XE}\|\sigma_X\otimes \sigma'_E)}$ is thus concave. The proof then follows from the definition of $I_{\alpha}^{*}(X;E)_{\sigma}$ given in \eqref{eq:Petz_Renyi_down_con} and the fact that pointwise infimum of concave functions is concave.   
\end{proof}
%% prints author names as small caps
%\renewcommand{\mkbibnamefirst}[1]{\textsc{#1}}
%\renewcommand{\mkbibnamelast}[1]{\textsc{#1}}
%\renewcommand{\mkbibnameprefix}[1]{\textsc{#1}}
%\renewcommand{\mkbibnameaffix}[1]{\textsc{#1}}

%\begin{comment}
%\newpage
\appendix
%\section{Appendix: Auxiliary Lemmas} 

%\textcolor{red}{We want to generalize \eqref{eq:BK} to positive semi-definite operators.}

%\end{comment}

%\newpage
\bibliographystyle{myIEEEtran}
\bibliography{reference.bib}

\end{document}